\documentclass[draftcls,onecolumn,peerreview]{IEEEtran}

\usepackage[final]{graphicx}
\usepackage[reqno]{amsmath}
\usepackage{amssymb}
\usepackage{cite}
\usepackage{enumerate}
\usepackage{algorithmic}
\usepackage{algorithm}
\usepackage{color}
\usepackage[normalem]{ulem}
\usepackage{psfrag}

\usepackage{booktabs}
\usepackage{multirow}
\usepackage{rotating}
\usepackage{bigstrut}

\def\expect{{\mathbb  E}}
\def\var{{\mathbb Var}}
\def\Pr{{\mathbb P}}

\def\nat{{\mathbb  N}}

\def\ind{{\bf 1}}

\def\12{\frac{1}{2}}

\newfont{\bbb}{msbm10 scaled 500}

\newfont{\bb}{msbm10 scaled 1100}

\newcommand{\hv}{{\bf h}}

\newcommand{\xv}{{\bf x}}
\newcommand{\yv}{{\bf y}}
\newcommand{\zv}{{\bf z}}

\newcommand{\Hm}{{\bf H}}

\newcommand{\Wm}{{\bf W}}

\newcommand{\Xm}{{\bf X}}
\newcommand{\Ym}{{\bf Y}}
\newcommand{\Zm}{{\bf Z}}

\newcommand{\Ac}{{\cal A}}
\newcommand{\Bc}{{\cal B}}

\newcommand{\Gc}{{\cal G}}

\newcommand{\Kc}{{\cal K}}

\newcommand{\Oc}{{\cal O}}
\newcommand{\Pc}{{\cal P}}

\newcommand{\Wc}{{\cal W}}

\newcommand{\Rate}{R}

\newtheorem{theorem}{Theorem}
\newtheorem{lemma}{Lemma}
\newtheorem{definition}{Definition}
\newtheorem{remark}{Remark}

\newtheorem{example}{Example}

\author{
Onur Gungor, Jian Tan, Can Emre Koksal, Hesham El-Gamal, Ness B. Shroff\\
\begin{tabular} {c}
\small Department of Electrical and Computer Engineering\\
\small The Ohio State University, Columbus, 43210\\
\end{tabular}
}

\title{Secrecy Outage Capacity of Fading Channels \footnote{This work was published in part
 in the Proceedings of INFOCOM 2010, San Diego, CA.}}

\begin{document}
\maketitle


\begin{abstract}

This paper considers point to point secure communication over flat fading channels
under an outage constraint. More specifically,
we extend the definition of outage capacity to account for the secrecy constraint and obtain  
sharp characterizations of the corresponding fundamental limits under two different assumptions on the transmitter CSI (Channel state information).
First, we find the outage secrecy capacity assuming that the transmitter has perfect knowledge of the legitimate and eavesdropper
channel gains. In this scenario, the capacity achieving scheme relies on opportunistically exchanging private keys between the legitimate nodes. These keys are stored in a key buffer and later used to secure delay sensitive data using the Vernam's one time pad technique. We then extend our results to the more practical scenario where the transmitter is assumed to know only the legitimate channel gain. Here, our achievability arguments rely on privacy amplification techniques to generate secret key bits. In the two cases, we also characterize the optimal power control policies which, interestingly, turn out to be a judicious combination of channel inversion and the optimal ergodic strategy. Finally, we analyze the effect of key buffer overflow on the overall outage probability. 


\end{abstract}


%


\section{Introduction}
\label{sec:Introduction}

Secure communication is a topic that is becoming increasingly
important thanks to the proliferation of wireless devices. Over the years, several secrecy protocols have been developed and incorporated in several wireless standards; e.g., the IEEE 802.11 specifications for Wi-Fi. However, as new schemes are being
developed, methods to counter the specific techniques also
appear. Breaking this cycle is critically dependent on the design of protocols that offer provable secrecy guarantees. The information theoretic secrecy paradigm adopted here, allows for a systematic approach for the design of low complexity and provable secrecy protocols that fully exploit the intrinsic properties of the wireless medium.

Most of the recent work on information theoretic secrecy is, arguably, inspired by Wyner's wiretap channel~\cite{Wyner}. In this setup, a passive eavesdropper listens to the communication between two legitimate nodes over a separate communication channel. While attempting to decipher the message, no limit is imposed on the computational resources available to the eavesdropper. This assumption led to defining 
{\bf perfect secrecy capacity} as the maximum achievable rate
subject to  zero mutual information rate between the transmitted message
and the signal received by the eavesdropper. In the additive Gaussian noise scenario~\cite{GaussianWiretap}, the perfect secrecy capacity turned out to be the difference between the capacities of the legitimate and eavesdropper channels. Therefore,
if the eavesdropper channel has a higher channel gain, information theoretic secure communication is not possible over the
main channel. Recent works have shown how to exploit multipath fading to avoid this limitation~\cite{Secrecy:08,KeysArq,khisti}. The basic idea is to opportunistically exploit the instants when the main channel enjoys a higher gain than the eavesdropper channel to exchange secure messages. This opportunistic secrecy approach was shown to achieve non-zero {\bf ergodic secrecy capacity} even when 
{\bf on average} the eavesdropper channel has favorable
conditions over the legitimate channel. Remarkably, this result still holds even when the channel state information of the eavesdropper channel is not available at the legitimate nodes.

The ergodic result in~\cite{Secrecy:08} applies only to delay tolerant traffic, e.g., file downloads. 
Early attempts at characterizing the delay limited secrecy capacity drew the negative conclusion that non-zero delay limited secrecy rates are not achievable, over almost all channel distributions, due to {\bf secrecy outage} events corresponding to the instants when the eavesdropper channel gain is larger than the main one~\cite{barros, liang}. Later, it was shown in~\cite{DelaySecrecy:09} that,
interestingly, a non-zero delay limited secrecy rate could be achieved by
introducing {\bf private key queues} at both the transmitter and the
receiver. These queues are used to store private key bits that are shared {\bf opportunistically} between the legitimate nodes when the main channel is more favorable than the one seen by the eavesdropper. These key bits are used later
to secure the delay sensitive data using the Vernam one time pad approach \cite{Shannon}.
Hence, secrecy outages are avoided by simply storing the secrecy generated previously, in the
form of key bits, and using them whenever the channel conditions are more advantageous for the eavesdropper. However, this work stopped short of proving sharp capacity results or deriving the corresponding optimal power control policies. These results can be recovered as special cases of the secrecy outage capacity and power control characterization obtained in the sequel. In particular, this work investigates the outage secrecy capacity of point-to-point block fading channels.  We first consider the scenario where perfect knowledge about the main and eavesdropper channels are available {\em a-priori} at the transmitter. The outage secrecy capacity and corresponding optimal power control policy is obtained and then the results are generalized to the more practical scenario where only the main channel state information (CSI) is available at the transmitter. 
Finally, the impact of the {\em private key queue} overflow on
secrecy outage probability is studied. 
Overall, our results reveals interesting structural insights on the optimal encoding and power control schemes as well as sharp characterizations of the fundamental limits on secure communication of delay sensitive traffic over fading channels.

The rest of this paper is organized as follows.
We formally introduce our system model in Section~\ref{s:sysmodel}. 
In Section~\ref{section:capacity}, we obtain the capacity results for the full and main CSI scenarios. 
The optimal power control policies, for both cases, are derived in
Section~\ref{section:power}. The effect of key buffer overflow on the outage probability is investigated 
in Section~\ref{s:finitebuffer}. We provide simulations to support our main results in Section~\ref{section:simulations}.
 Finally,
Section~\ref{s:conclusion} offers some concluding remarks. To enhance the flow of the paper, 
the proofs are collected in the Appendices. 

\section{System Model}\label{s:sysmodel}
We study a point-to-point wireless communication link, 
in which a transmitter is trying to send information to a legitimate receiver,
under the presence of a passive eavesdropper.
We divide time into discrete slots, where blocks are formed by $N$ channel uses, and 
 $B$ blocks combine to form a super-block. Let the communication period consist of $S$ super-blocks. 
We use the notation $(s,b)$ to denote the $b^\text{th}$ block in the $s^\text{th}$ super-block.
We adopt a block fading channel model, in which the channel is assumed to be constant over a block,
and changes randomly from one block to the next.
Within each block $(s,b)$, the observed
signals at the receiver and at the eavesdropper are:
\begin{align*}
\Ym(s,b) &= G_m(s,b)\Xm(s,b)+\Wm_m(s,b)
\end{align*}
and
\begin{align*}
\Zm(s,b) &= G_e(s,b)\Xm(s,b)+\Wm_e(s,b),
\end{align*}
respectively, where $\Xm(s,b)\in {\mathbb C}^N$ is the transmitted
signal, $\Ym(s,b)\in {\mathbb C}^N$ is the received signal by the
legitimate receiver, and $\Zm(s,b)\in {\mathbb C}^N$ is the received
signal by the eavesdropper. 
$\Wm_m(s,b)$ and $\Wm_e(s,b)$ are independent noise vectors, whose elements are drawn from
standard complex normal distribution.
We assume that the channel gains of the main channel $G_m(s,b)$ and the eavesdropper channel
$G_e(s,b)$ are i.i.d. complex random variables.
 The power gains of the
fading channels are denoted by $H_m(s,b)=|G_m(s,b)|^2$ and
$H_e(s,b)=|G_e(s,b)|^2$.  We sometimes use the vector notation $\Hm(\cdot) = [H_m(\cdot)~H_e(\cdot)]$ 
for simplicity, and also use the notation $\Hm^{s,b} = \{\Hm\}_{s'=1,b'=1}^{s,b}$ to denote the set of
channel gains $\Hm(s',b')$ observed until block $(s,b)$. We use similar notation for other signals as well,
and denote the sample realization sequences with lowercase letters.
We assume that the probability density function of instantaneous channel gains, 
denoted as $f(\hv)$, is well defined,
and is known by all parties. 
We define channel state information (CSI) as one's knowledge of the instantaneous
channel gains.
We define {\em full transmitter CSI} as the case in which the transmitter
has full causal knowledge of the main and eavesdropper channel gains.
We define {\em main transmitter CSI} as the case in which
that the transmitter only knows the CSI of the legitimate receiver. 
In both cases, the eavesdropper has complete knowledge of both 
the main and the eavesdropper channels.
Let $P(s,b)$ denote the power allocated at block $(s,b)$.
We consider a long term power constraint (or average power constraint) such that,
\begin{align}
\limsup_{S,B \to \infty}\frac{1}{SB}\sum_{s=1}^S\sum_{b=1}^B
P(s,b) \leq P_{\text{avg}} \label{avgpowerconstraint}
\end{align}
for some $P_{\text{avg}}>0$.


Let $\{W(s,b)\}_{s=1,b=1}^{S,B}$ denote the set of messages to be transmitted with a delay constraint. 
$W(s,b)$ becomes available at the transmitter at the beginning of block $(s,b)$, and needs to be securely
communicated to the legitimate receiver at the end of that particular block. We consider the problem of constructing
$(2^{N\Rate},N)$ codes to communicate message packets $W(s,b)\in\{1,\cdots,2^{N\Rate}\}$ of equal size, which consists of:
\begin{enumerate}
\item A stochastic encoder 
that maps $(w(s,b)$, $\xv^{s,b-1})$ to $\xv(s,b)$ based on the available CSI, where $\xv^{s,b-1}$ summarizes the
previously transmitted signals\footnote{An exception is for $b=1$, in which case the previous signals 
 are summarized by $\xv^{s-1,B}$.}, and
\item A decoding function 
that maps $\yv^{s,b}$ to $\hat{w}(s,b)$ at the legitimate receiver.
\end{enumerate}
Note that we consider the current block $\xv(s,b)$ to be a function of the past blocks $\xv^{s,b-1}$ as well. 
This kind of generality allows us to store shared randomness 
 to be exploited 
in the future to increase the achievable secrecy rate.

Define the error event with parameter $\delta$ at block $(s,b)$ as 
$$ E(s,b,\delta)= \big{\{}\hat{W}(s,b)\neq W(s,b)\big{\}} \cup \big{\{}\frac{1}{N}\|\Xm(s,b)\|^2 > P(s,b)+\delta \big{\}}  $$
which occurs either when the decoder makes an error, or 
when the power expended is greater than $P(s,b)+\delta$.
The equivocation rate at the eavesdropper is defined as
 the entropy rate of the message at block $(s,b)$, 
 conditioned on the received signal by the eavesdropper
 during the transmission period, and available eavesdropper CSI, which is equal to
$\frac{1}{N} H(W(s,b)|\Zm^{SB},\hv^{SB})$.
The secrecy outage event at rate $\Rate$ with parameter $\delta$
at block $(s,b)$ is defined as 
\begin{align}
{\Oc}_{\text{sec}}(s,b,\Rate,\delta) = {\Oc}_{\text{eq}}(s,b,\Rate,\delta) \cup {\Oc}_{\text{ch}}(s,b,\Rate) \label{eq:secrecyoutage}
\end{align}
where the equivocation outage $${\Oc}_{\text{eq}}(s,b,\Rate,\delta) = \left\{ \frac{1}{N} H(W(s,b)|\Zm^{SB},\hv^{SB}) < \Rate - \delta \right\}$$ occurs if the equivocation rate at block $(s,b)$ is less than $\Rate-\delta$,
and channel outage $${\Oc}_{\text{ch}}(s,b,\Rate) =\left\{ \frac{1}{N}I\left(\Xm(s,b);\Ym(s,b)\right)< \Rate\right\} $$ occurs if channel at block $(s,b)$ is unsuitable for reliable transmission at rate $\Rate$.
Defining $\bar{\Oc}_{\text{sec}}(\cdot)$ as the complement of the event ${\Oc}_{\text{sec}}(\cdot)$,
we now characterize the notion of $\epsilon$-achievable secrecy capacity. 
\begin{definition} 
Rate $\Rate$ is achievable securely with at most $\epsilon$ probability of secrecy outage if, 
for any fixed $\delta>0$, there exist $S,B$ and $N$ large enough such that the conditions
\begin{align}
\Pr(E(s,b,\delta)|\bar{\Oc}_{\text{sec}}(s,b,\Rate,\delta)) < \delta \label{errorconstraint}\\
\Pr ({\Oc}_{\text{sec}}(s,b,\Rate,\delta)) < \epsilon+\delta  \label{outage} 
\end{align}
are satisfied for all $(s,b)$, $s \neq 1$.
\end{definition}

We call such $\Rate$ an  $\epsilon$-achievable secrecy rate. Note that the security constraints are not imposed on the first super-block.
\begin{definition}
The $\epsilon$-achievable secrecy capacity
is the supremum of $\epsilon$-achievable secrecy rates $\Rate$.
\end{definition}
\begin{remark}
The notion of secrecy outage was previous defined and used in \cite{liang,barros}. However, those works did not
consider the technique of storing shared randomness for future use, and in that case, secrecy outage 
depends only on the instantaneous channel states. In our case, secrecy outage depends on previous channel states as well.
Note that we do not impose a secrecy outage constraint on the first superblock ($s=1$). 
We refer to the first superblock as an initialization phase
used to generate initial common randomness between the legitimate
nodes. Note that this phase only needs to appear \emph{once} in the communication lifetime of that link. In other words, when a session (which consists of $S$ superblocks) between the associated nodes is over, they would have 
sufficient number of common key bits for the subsequent session,
and would not need to initiate the initialization step again.
\end{remark}
\section{Capacity Results}\label{section:capacity}
 In this section, we investigate this capacity 
under two different cases; full CSI and main CSI at the transmitter.
Before giving the capacity results, we define the following quantities.
For a given power allocation function $P(s,b)$, let $R_m(s,b)$ and $R_s(s,b)$ be as 
follows,
\begin{align}
  R_m(s,b)  = & \log(1+P(s,b)H_m(s,b)) \label{eqn:Rm(t)}\\
  R_s(s,b)  = & [\log(1+P(s,b)H_m(s,b))-\log(1+P(s,b)H_e(s,b))]^+ \label{eqn:Rs(t)}
\end{align} 
where $[\cdot]^+= \max(\cdot,0)$. Note that, $R_m(\cdot)$ is the supremum of achievable main channel rates,
 without the secrecy constraint.
Also, $R_s(\cdot)$
is the non-negative difference between main channel 
and eavesdropper channel's supremum achievable rates. 
We show in capacity proofs that the outage capacity achieving power allocation functions lie in
the space of stationary power allocation functions that are functions of instantaneous transmitter CSI.
Hence for \textbf{full CSI},
we constrain ourselves to the set $\Pc$ of stationary power allocation policies that are
 functions of $\hv(s,b)=[h_m(s,b)~h_e(s,b)]$.
For simplicity, we drop the block index $(s,b)$, and use the notation
$P(\hv)$ for the stationary power allocation policy.
Similarly, with \textbf{main CSI}
 we consider the power allocation policies that are
 functions of $h_m(s,b)$, and
use the notation $P(h_m)$ for the stationary power allocation policy.
In both cases, since the secrecy rate $R_s(s,b)$, and the main channel rate $R_m(s,b)$ are completely determined 
by the stationary power allocation functions $P(\cdot)$ and channel gains $\hv$, we will interchangeably
 use the notations $R_s(s,b) \equiv R_s(\hv,P)$ and $R_m(s,b) \equiv R_m(\hv,P)$.
 
\subsection{Full CSI}\label{s:FullCSI}
\begin{theorem}\label{FullCSICapacity}
Let the transmitter have full CSI.
Then, for any $\epsilon$, $0\leq \epsilon < 1$, the 
$\epsilon$-achievable secrecy capacity
is identical to
\begin{align}
C_F^{\epsilon} = \max_{P(\hv) \in \cal P'}\frac{\expect[R_s(\Hm,P)]}{1-\epsilon} \label{SecrecyCapacity}
\end{align}
where the set $\cal P'  \subseteq P$ consists of power control policies $P(\hv)$ that satisfies the following conditions.
\begin{align}
 \Pr\left(R_m(\Hm,P) < \frac{\expect[R_s(\Hm,P)]}{1-\epsilon}\right) &\leq \epsilon  \label{eq:OutageConstraint}\\
\expect [P(\Hm)] &\leq P_{\text{avg}} \label{eq:PowerConstraint2}
\end{align}
\end{theorem}

A detailed proof of achievability and converse part is provided in Appendix~\ref{a:FullCSI}.  
Here, we briefly justify the result. 
For a given $P(\hv)$, 
$R_s(\hv,P)$ the supremum of the secret key generation rates within a block that
experiences channel gains $\hv$ \cite{GaussianWiretap}.
This implies that the expected achievable secrecy rate 
\cite{Secrecy:08} is $\expect[R_s(\Hm,P)]$ without the outage constraint. 
With the outage constraint, the fluctuations of $R_s(\Hm,P)$ due to fading are unacceptable,
since $R_s(\Hm,P)$ can go below the desired rate when the channel conditions are 
unfavorable (e.g., when $H_m<H_e$, $R_s(\Hm,P)=0$).
 Hence, we utilize secret key buffers to smoothen out these
fluctuations to provide secrecy rate of $\expect[R_s(\Hm,P)]$ at each block.
The generated secrecy is stored in secret key buffers of both the transmitter and receiver, and is utilized
 to secure data of same size using Vernam's one-time pad technique.
With the allowable amount of secrecy outages, this rate goes up to $\expect[R_s(\Hm,P)]/(1-\epsilon)$.
The channel outage constraint \eqref{eq:OutageConstraint} on the other hand
is a necessary condition to satisfy the secrecy outage constraint in \eqref{outage} due to \eqref{eq:secrecyoutage}.
\begin{example}\label{e:FullCSI}
Consider a four state system, where $H_m$ and $H_e$ takes values from the set $\{1,10\}$ and the 
joint probabilities are as given in Table~\ref{ex1:Pr}.
Let the average power constraint be $P_{\text{avg}} =0.5$, and there is no
 power control, i.e., $P(\hv) = P_{\text{avg}}$ $\forall \hv$. The achievable
instantaneous secrecy rate at each state is given in Table~\ref{ex1:Rs}. 
According to the pessimistic result in [6,8],
any non-zero rate cannot be achieved with a secrecy outage probability $\epsilon<0.6$ in this case. 
However, according to Theorem 1, rate $R = \frac{0.8}{1-\epsilon}$ can be achieved with $\epsilon$
secrecy outage probability\footnote{Although Theorem 1 is stated for the case where random vector $\Hm$
is continuous, the result
similarly applies to discrete $\Hm$ as well.}, since $\expect[R_s(\Hm,P_{\text{avg}})]=0.8$.
A sample path is provided for both schemes in Figure~\ref{fig:sample_path}, and it is shown how
our scheme avoids secrecy outage in the second block.
\end{example}
\begin{table}[ht]
\begin{minipage}[b]{0.5\linewidth}\centering
    \caption{$\Pr(\hv)$}
    \label{ex1:Pr}
    \begin{tabular}{c | c | c}
    \toprule
    $\downarrow h_m$ \textbackslash ~ $h_e \rightarrow$ & 1     & 10 \\
    \midrule
    1     & 0.1   & 0.1 \\
    10    & 0.4   & 0.4 \\
    \bottomrule
    \end{tabular}%
\end{minipage}
\hspace{0.5cm}
\begin{minipage}[b]{0.5\linewidth}\centering
    \caption{$R_s(\hv,P_{\text{avg}})$}
    \label{ex1:Rs}
    \begin{tabular}{c | c | c}
    \toprule
    $\downarrow h_m$ \textbackslash ~ $h_e \rightarrow$ & 1     & 10 \\
    \midrule
    1     & 0     & 0 \\
    10    & 2     & 0 \\
    \bottomrule
    \end{tabular}%
\end{minipage}
\end{table}

 \begin{figure}[!h]
\begin{center}
\psfrag{t=1}[][][0.8]{\hspace{3cm} $\text{block 1},~\hv = [10~1]$}
\psfrag{t=2}[][][0.8]{\hspace{3cm} $\text{block 2},~\hv = [10~ 10]$}
\psfrag{Rm=7Re=8Rs=0}[][][0.8]{\hspace{1cm} $R_m = 2.58,~R_s = 0$}
\psfrag{Rm=5Re=1Rs=4}[][][0.8]{\hspace{1cm} $R_m = 2.58,~R_s=2$}
\psfrag{Rs=0}[][][0.6]{\hspace{0.5cm}$R_s = 0$}
\psfrag{Rs=2}[][][0.6]{\hspace{0.5cm}$R_s = 2$}
\includegraphics[height=1.8in]{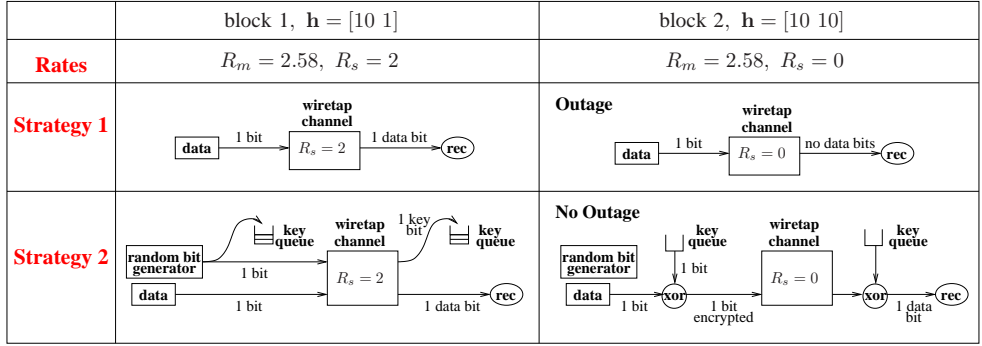}
\caption{A sample path. With strategy 2, secrecy outage can be avoided
for block $t=2$ via the use of key bits.}
\label{fig:sample_path}
\end{center}
\end{figure}
\subsection{Main CSI}\label{s:MainCSI}


\begin{theorem}\label{MainCSICapacity}
Let the transmitter have main CSI.
Then, for any $\epsilon$, $0\leq \epsilon < 1$, the $\epsilon$-achievable secrecy capacity 
is identical to
\begin{align}
C_M^{\epsilon} = \max_{P(h_m) \in \cal P''}\frac{\expect[R_s(\Hm,P)]}{1-\epsilon} \label{SecrecyCapacity2}
\end{align}
where the set $\cal P''\subseteq P$ consists of power control policies $P(h_m)$ that satisfies the following conditions.
\begin{align}
 \Pr\left(R_m(\Hm,P) < \frac{\expect[R_s(\Hm,P)]}{1-\epsilon}\right) &\leq \epsilon  \label{eq:OutageConstraint2}\\
\expect [P(H_m)] &\leq P_{\text{avg}} \label{eq:PowerConstraint3}
\end{align}
\end{theorem}
Although the problems \eqref{SecrecyCapacity}-\eqref{eq:PowerConstraint2} and 
\eqref{SecrecyCapacity2}-\eqref{eq:PowerConstraint3}
 are of the same form, due to the
absence of eavesdropper CSI, the maximization in this case is over power allocation 
functions $\Pc''$ that depend on
the main channel state only. Hence, $C_M^{\epsilon}\leq C_F^{\epsilon}$.
A detailed proof of achievability and converse is provided in Appendix~\ref{a:MainCSI}. 
As in the full CSI case, our achievable scheme uses similar key buffers and  Vernam's one time pad technique to
secure the message. The main difference 
is the generation of secret key bits. 
Due to the lack of knowledge of $H_e(s,b)$
at the transmitter, secret key bits cannot be generated within a block.
Instead, using the statistical knowledge of $H_e(s,b)$, 
we generate keys over a super-block.
Roughly, over a superblock the receiver can reliably obtain
$NB\expect[R_m(\Hm,P)]$ bits
of information, while the eavesdropper can obtain $NB\expect[R_m(\Hm,P) - R_s(\Hm,P)]$ bits of information. From
privacy amplification arguments \cite{Bennett:95}, $NB\expect[R_s(\Hm,P)]$ bits of secret key
 can be extracted by using a universal hash function.

Now, we show that power allocation policy has minimal impact on
the performance in the high power regime.
\begin{theorem}\label{t:highpowercapacity}
For any $\epsilon>0$, the $\epsilon$-achievable secrecy capacities with full CSI and main CSI converge to
the same value
\begin{align}
\lim_{P_{\text{avg}} \to \infty}C_F^{\epsilon} = \lim_{P_{\text{avg}} \to \infty} C_M^{\epsilon} = \frac{\expect_{H_m>He}\log\left(H_m/H_e\right)}{(1-\epsilon)}
\end{align}
\end{theorem}
\begin{proof}
For $\hv \equiv [h_m ~ h_e]$ such that $h_m > h_e$, we can see from \eqref{eqn:Rs(t)} that
$\lim_{P(\hv)\to \infty} R_s(\hv,P) = \log\left(\frac{h_m}{h_e}\right)$, 
and for $h_m \leq h_e$, $R_s(\hv,P) = 0$. 
Furthermore, for $h_m>0$, we can see from \eqref{eqn:Rm(t)} that
$\lim_{P(\hv)\to \infty} R_m(\hv,P) = \infty $.  Let $P(\hv)= P_{\text{avg}}$ (no power
control), which does not require any CSI. Then, we get
$$\lim_{P_{\text{avg}} \to \infty} \expect[R_s(\Hm,P)] =    \expect_{H_m>He}\log\left(H_m/H_e\right) < \infty $$
Combining the last 2 equations, we get
$$  \lim_{P_{\text{avg}} \to \infty}\Pr\left(R_m(\Hm,P)< \frac{\expect[R_s(\Hm,P)]}{1-\epsilon}\right) = \Pr(H_m = 0)$$
and  $ \Pr(H_m = 0) =0 $, since probability density function of $\Hm$ is well defined. Hence, channel outage constraints
\eqref{eq:OutageConstraint} and \eqref{eq:OutageConstraint2} are not active in the high power regime.
Therefore, $P(\hv) \in \Pc'$,  and $P(\hv) \in \Pc''$.
From \eqref{SecrecyCapacity}-\eqref{eq:PowerConstraint2} and \eqref{SecrecyCapacity2}-\eqref{eq:PowerConstraint3},
we conclude that
$C_F=C_M=\expect_{H_m>He}\log\left(H_m/H_e\right)/(1-\epsilon)$.
\end{proof}
Our simulation results also illustrate that the power allocation policy has minimal impact on the importance in
the high power regime. On the other hand, when the average power is limited, the optimality of the power allocation
function is of critical importance, which is the focus of the following section.
\section{Optimal Power Allocation Strategy}\label{section:power}

\subsection{Full CSI}\label{s:FullCSIPower}
The optimal power control strategy, $P^{\ast}(\hv)$ is the stationary strategy that 
solves the optimization problem \eqref{SecrecyCapacity}-\eqref{eq:PowerConstraint2}.
In this section, we will show that $P^{\ast}(\hv)$ is a time-sharing between the channel inversion power policy, and the secure waterfilling
policy.
We first introduce the channel inversion power policy, $P_{\text{inv}}(\hv,\Rate)$,
which is the \emph{minimum} required power to maintain main channel rate of $\Rate$.
For $\hv = [h_m ~ h_e]$, 
\begin{align}
P_{\text{inv}}(\hv,\Rate)&=\frac{2^{\Rate}-1}{h_m} \label{ChannelInversion}
\end{align}
Next we introduce $P_{\text{wf}}(\hv, \lambda)$,
\begin{align}
P_{\text{wf}}(\hv, \lambda)&=\frac{1}{2}{\Big
[}\sqrt{\left(\frac{1}{h_e}-\frac{1}{h_m}\right)^2+\frac{4}{\lambda}\left(\frac{1}{h_e}-\frac{1}{h_m}\right)}
 -\left(\frac{1}{h_e}+\frac{1}{h_m}\right){\Big ]}^+, \label{Waterfilling} 
\end{align}
We call it the 'secure waterfilling' power policy because it maximizes the ergodic secrecy rate without any outage constraint,
and resembles the 'waterfilling' power control policy.
Here, the parameter $\lambda$ determines the power expended on average. 
Now, let us define a time-sharing region
\begin{align}
\Gc(\lambda,k)=\left\{ \hv: \left[R_s(\hv,P_{\text{inv}})-R_s(\hv,P_{\text{wf}})\right]^+ 
-\lambda \left[P_{\text{inv}}(\hv,b)-P_{\text{wf}}(\hv,\lambda)\right]^+ \geq k \right\} \label{setG}
\end{align}
which is a function of parameters $\lambda$ and $k$.
\begin{theorem}\label{t:FullCSIPower}
$P^{\ast}(\hv)$ is the unique solution to 
\begin{align}                          
                    P^{\ast}(\hv)=&P_{\text{wf}}(\hv,\lambda^{\ast})+ \ind\left(\hv \in \Gc(\lambda^{\ast},k^{\ast})\right) 
                          \left(P_{\text{inv}}(\hv,C_F^{\epsilon})-P_{\text{wf}}(\hv,\lambda^{\ast})\right)^+ \label{optimalPower}  \\
                     \text{subject to: }& k^{\ast}\leq 0, \lambda^{\ast} >0 \nonumber\\                         
                    & C_F^{\epsilon}=\expect[R_s(\Hm,P^{\ast})]/(1-\epsilon)  \label{optimalPower:2}\\
                    & \Pr(\Hm \in \Gc(\lambda^{\ast},k^{\ast})) = 1-\epsilon  \label{optimalPower:3}\\
					& \expect[P^{\ast}(\Hm)] = P_{\text{avg}} \label{optimalPower:4}
\end{align}
where $\expect[R_s(\Hm,P^{\ast})]$ is the expected secrecy rate under the power allocation policy $P^{\ast}(\hv)$. 
\end{theorem}
\begin{proof} 
Define a sub-problem
\begin{align}
\expect[R_s(\Hm,P^{\Rate})]=&\max_{P(\hv)}\expect \left[  R_s(\Hm,P)   \right] \label{eq:PowerControl}\\
  \textrm{subject to: }& P(\hv) \geq 0,~\forall \hv \nonumber\\
&  \expect [P(\Hm)] \leq P_{\text{avg}},  \label{powerconstraint}\\  
&  \Pr\left(R_m(\Hm,P) < \Rate \right) \leq \epsilon  \label{serviceconstraint}
\end{align}
Let $P^{\Rate}(\hv)$ be the power allocation function that solves this sub-problem. 
Note that for $\Rate = \expect[R_s(\Hm,P^{\Rate})]/(1-\epsilon)$,
 this problem is identical to  \eqref{SecrecyCapacity}-\eqref{eq:PowerConstraint2}, hence
giving us $P^{\ast}(\hv)$. We will prove the existence and uniqueness of such $\Rate$.
\begin{lemma}\label{bmaxlemma}
There exists a unique $\Rate_{\text{max}}>0$ such that
the sub-problem \eqref{eq:PowerControl}-\eqref{serviceconstraint} has a solution for all $\Rate \leq \Rate_{\max}$, which is found by solving
\begin{eqnarray}\label{feasiblePwr}
P_{\text{avg}} = \int_{h_m\geq c}P_{\text{inv}}(\hv,\Rate_{\max})f(\hv)d\hv
\end{eqnarray}
for $\hv \equiv [h_m~h_e]$, where the constant $c$ is chosen such that $\Pr(H_m \leq c)=\epsilon$. 
\end{lemma}
Proof is provided in Appendix~\ref{a:bmaxlemma}.
\begin{lemma}\label{l:FullCSIPower}
For any $\Rate \leq \Rate_{\max}$, 
\begin{align*}
P^{R}(\hv)=&P_{\text{wf}}(\hv,\lambda)+ \ind\left(\hv \in \Gc(\lambda,k)\right) 
                          \left(P_{\text{inv}}(\hv,\Rate)-P_{\text{wf}}(\hv,\lambda)\right)^+
\end{align*}                          
where $k\in(-\infty, 0]$
 and $\lambda \in (0,+\infty)$
are parameters that satisfy \eqref{powerconstraint} and \eqref{serviceconstraint} with equality.
\end{lemma}
Proof is provided in Appendix~\ref{a:FullCSIPower}.
It is left to show there exists a unique $\Rate$ that satisfies
 $\Rate = \expect[R_s(\Hm,P^{\Rate})]/(1-\epsilon)$.
\begin{lemma} \label{l:Rb}
$\expect[R_s(\Hm,P^{\Rate})]$ is a continuous non-increasing function of $\Rate$.
\end{lemma}
Proof is provided in Appendix~\ref{al:Rb}.
\begin{lemma}\label{fullCSIuniquepwr}
There exists a unique $\Rate$, $0\leq \Rate \leq \Rate_{\max}$, which satisfies $\Rate=\expect[R_s(\Hm,P^{\Rate})]/(1-\epsilon)$.
\end{lemma}
Proof is provided in Appendix~\ref{a:fullCSIuniquepwr}.
This concludes the proof of the theorem.
\end{proof}
Due to \eqref{optimalPower}, the optimal power allocation function is a time-sharing 
between the channel allocation power allocation function
and secure waterfilling; a balance between avoiding channel outages, hence secrecy outages, and maximizing the expected secrecy rate.
The time sharing region $\Gc(\lambda,k)$ determines the instants $\hv$, for which
avoiding channel outages are guaranteed  through the choice of $P(\hv) = \max(P_{\text{inv}}(\hv,\Rate),P_{\text{wf}}(\hv,\lambda))$.
 \eqref{optimalPower:3} ensures that channel outage probability is at most $\epsilon$, and \eqref{optimalPower:4}
ensures that average power constraint is met with equality.
\eqref{optimalPower:2}, on the other hand, is an immediate consequence of \eqref{SecrecyCapacity}.

Note that, an extreme case is $P^{\ast}(\hv) = P_{\text{wf}}(\hv,\lambda^{\ast})$ $\forall \hv$, which occurs when $P_{\text{inv}}(\hv,\Rate)\leq P_{\text{wf}}(\hv,\lambda^{\ast})$
for any $\hv \in \Gc(\lambda^{\ast},k^{\ast})$, which translates into the fact that the secure waterfilling solution
itself satisfies the channel outage probability in \eqref{eq:OutageConstraint}. However,
that the other extreme ($P^{\ast}(\hv) = P_{\text{inv}}(\hv,\Rate^{\ast})$ $\forall \hv$) cannot occur for any non-zero $\epsilon$ due to  \eqref{optimalPower}.
The parameter $C_F^{\epsilon}$ can be found graphically as shown in Figure~\ref{fig:FullCSIPower},
 by plotting $\expect[R_s(\Hm,P^{\Rate})]$ and
and $(1-\epsilon)\Rate$ as a function of $\Rate$. The abcissa of the unique intersection point is $\Rate=C_F^{\epsilon}$.
\noindent \begin{figure}[h]
\centerline{\includegraphics[width=10cm]{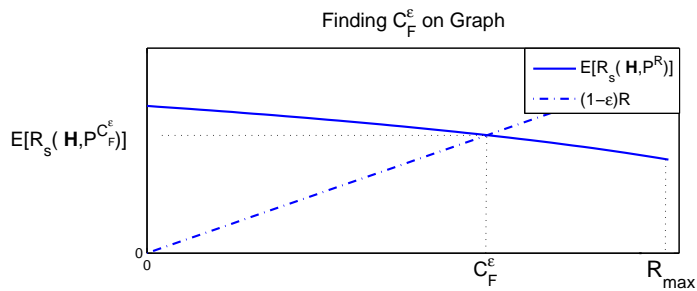}}
    \caption{Finding $C_F^{\epsilon}$ with graphical approach}
    \label{fig:FullCSIPower}
\end{figure}
\begin{example}
Consider the same system model in Example~\ref{e:FullCSI}. We have found that for $R=\frac{0.8}{1-\epsilon}$ bits/channel use is
achievable with $\epsilon$ probability of secrecy outage with no power control, i.e., $P(\hv) =0.5$ $\forall \hv$.
Let $\epsilon =0.2$,  we will see if we can do better than $R=1$ with power control.
Solving the problem \eqref{optimalPower}-\eqref{optimalPower:4},
we can see that\footnote{Although Theorem~\ref{t:FullCSIPower} assumes $\Hm$ is a continuous random vector,
the results similarly hold for the discrete case as well.} the time-sharing, and
power expended in each state are as given in Tables~\ref{ex2:Gc} and \ref{ex2:P}. For $\hv\equiv[h_m ~ h_e]=[10 ~ 1]$, i.e.,
the legitimate channel has a better gain, secure waterfilling
is used and when $\hv=[10 ~ 10]$,
secret key bits cannot be generated, but
channel inversion is used to guarantee a main channel rate of $\Rate$, which is secured by the excess keys generated during the state $\hv=[10 ~ 1]$.
As a result, we can see that a rate of $C_F^{0.2}=1.26$ bits/ per channel use is achievable, which corresponds to $26 \%$ increase with
respect to no power control. As mentioned in Theorem~\ref{t:highpowercapacity}, this gain diminishes at the high power regime, i.e., when $P_{\text{avg}}\to \infty$.
\end{example}
\begin{table}[ht]
\begin{minipage}[b]{0.5\linewidth}\centering
    \caption{Time sharing regions}
    \label{ex2:Gc}
    \begin{tabular}{c | c | c}
    \toprule
    $\downarrow h_m$ \textbackslash ~ $h_e \rightarrow$ & 1     & 10 \\
    \midrule
    1     & wf    & wf \\
    10    & wf    & inv \\
    \bottomrule
    \end{tabular}%
\end{minipage}
\hspace{0.5cm}
\begin{minipage}[b]{0.5\linewidth}\centering
    \caption{$P^{\ast}(\hv)$}
    \label{ex2:P}
    \begin{tabular}{c | c | c}
    \toprule
    $\downarrow h_m$ \textbackslash ~ $h_e \rightarrow$ & 1     & 10 \\
    \midrule
    1     & 0     & 0 \\
    10    & 1.11  & 0.14 \\
    \bottomrule
    \end{tabular}%
\end{minipage}
\end{table}
\subsection{Main CSI}
Here, we  find the optimal power control strategy $P^*(h_m)$, which solves the optimization 
problem \eqref{SecrecyCapacity2}-\eqref{eq:PowerConstraint3}. 
Let us define $P_{w}(h_m,\lambda)$ as the maximum of $0$, and the solution of the following equation
\begin{align}
\frac{\partial \expect[R_s(\Hm,P)]}{\partial P(h_m)}= 
\frac{h_m \Pr (h_e \leq h_m)}{1+h_mP(h_m)}-\int_0^{h_m}\left( \frac{h_e}{1+h_eP(h_m)} \right)f(h_e)dh_e - \lambda = 0
\end{align}
$P_{w}(h_m,\lambda)$ will replace $P_{\text{wf}}(\hv,\lambda)$ in the full CSI case.
\begin{theorem} 
$P^{\ast}(h_m)$ is the unique solution to
\begin{align}           
                   P^{\ast}(h_m) = & P_w(h_m , \lambda^{\ast}) +  \ind(h_m \geq c) \left( P_{\text{inv}}(h_m,C_M^{\epsilon}) 
                             - P_{w}(h_m,\lambda^{\ast}) \right)^+ \label{maincsirateb}  \\
                     \text{subject to: }&\lambda^{\ast}>0        \nonumber\\
                    & C_M^{\epsilon}=\expect[R_s(\Hm,P^{\ast})]/(1-\epsilon)  \label{MoptimalPower:2}\\
                    & \Pr(H_m \geq c) = 1-\epsilon  \label{MoptimalPower:3}\\
					& \expect[P^{\ast}(H_m)] = P_{\text{avg}} \label{MoptimalPower:4}
\end{align}
where $\expect[R_s(\Hm,P^{\ast})]$ is the expected secrecy rate under the power allocation policy $P^{\ast}(h_m)$. 
\end{theorem}
\begin{proof}
The proof follows the approach in Full CSI case, hence we omit the details for brevity.
Define the sub-problem
\begin{align}
\expect[R_s(\Hm,P^{\Rate} )] = &\max_{P(h_m)}\expect \left[  R_s(\Hm,P)   \right] \label{eq:MainPowerControl}\\
  \textrm{subject to: }& P(h_m) \geq 0,~ \forall h_m \nonumber\\
&  \expect [P(H_m)] \leq P_{\text{avg}},  \label{mainpowerconstraint}\\  
&  \Pr\left(R_m(\Hm,P) < \Rate \right) \leq \epsilon  \label{mainserviceconstraint}
\end{align}
Let $P^{\Rate}(h_m)$ be the power allocation function that solves this sub-problem.
Lemmas \ref{bmaxlemma} and \ref{fullCSIuniquepwr} also hold in this case.
The only difference is the following lemma, which replaces Lemma~\ref{l:FullCSIPower} in Full CSI.
\begin{lemma}\label{t:MainCSIPower}
For any $R\leq R_{\max}$,
\begin{align*}
P^{\Rate}(h_m)=P_{w}(h_m,\lambda)+ \ind(h_m > c) 
                          \left(P_{\text{inv}}(h_m,\Rate)-P_{w}(h_m,\lambda)\right)^+
\end{align*}                          
where $c$ is a constant that
satisfies $\Pr(H_m \geq c) = 1-\epsilon$, and
 $\lambda \in (0,+\infty)$ is a constant that satisfies \eqref{mainpowerconstraint} with equality.
\end{lemma}
The proof is similar to the proof of Lemma~\ref{l:FullCSIPower}, and is provided in Appendix~\ref{a:MainCSIPower}.
\end{proof}
The graphical solution in Figure~\ref{fig:FullCSIPower} to find $C_F^{\epsilon}$ also generalizes to the main CSI case.


\section{Sizing the Key Buffer}\label{s:finitebuffer}

The proofs of the capacity results of Section~\ref{section:capacity}
assume availability of \emph{infinite size} secret
 key buffers at the transmitter and receiver,
which  mitigate the effect of fluctuations in the achievable secret key bit rate due to fading. 
Finite-sized buffers, on the other hand will lead to a higher secrecy outage probability due to wasted
key bits by the key buffer overflows. 
We revisit the full CSI problem, and we consider this problem at `packet' level, where we assume a
packet is of fixed size of $N$ bits. We will prove the following result.
\begin{theorem}\label{eq:BuffervsOutage}
Let $\epsilon'>\epsilon$. Let $M_{C_F^{\epsilon}}(\epsilon')$ be the buffer size (in terms of packets) 
sufficient to achieve rate $C_F^{\epsilon}$ 
with at most $\epsilon'$
probability of secrecy outage. Then,
\begin{align}
\lim_{\epsilon'\searrow \epsilon} 
\frac{M_{C_F^{\epsilon}}(\epsilon')-C_F^{\epsilon}} {\frac{ \var [R_s (\Hm,P^{C_F^{\epsilon}})]
  + (C_F^{\epsilon})^2\epsilon(1-\epsilon)}{(\epsilon'-\epsilon)C_F^{\epsilon}}
\log\left( \frac{\var[R_s(\Hm,P^{C_F^{\epsilon}})] + (C_F^{\epsilon})^2\epsilon(1-\epsilon)} {(\epsilon'-\epsilon)^2 
C_F^{\epsilon}} \right)} \leq 1 \label{c:asymptoticbuffer}
\end{align}
\end{theorem} 
Before providing the proof, we first interpret this result. If buffer size is infinite, 
we can achieve rate $C_F^{\epsilon}$
with $\epsilon$ probability of secrecy outage. With finite buffer, we can achieve the same rate with $\epsilon'$ probability of secrecy outage.
Considering this difference to be the price that we have to pay due to the finiteness of the buffer, 
 we can see that the buffer
size required scales with $\text{O}\left(\frac{1}{\epsilon'-\epsilon} \log\frac{1}{\epsilon'-\epsilon} \right)$, as $\epsilon'-\epsilon \to 0$. 
\begin{proof}
Achievability follows from simple modifications to the capacity achieving scheme
described in Appendix~\ref{a:FullCSI}. 
We will first study the key queue dynamics, then using the heavy traffic limits, we provide an upper bound
to the key loss ratio due to buffer overflows. Then, we relate key loss ratio to the secrecy outage probability,
and conclude the proof.

For the key queue dynamics, we use a single index $t$ to denote the time index 
instead of the double index $(s,b)$,
where $t = sB+b$. 
We consider transmission at outage secrecy rate of $\Rate$, and use power 
allocation function $P^{\Rate}(\hv)$, which solves
the problem \eqref{eq:PowerControl}-\eqref{serviceconstraint}.
Let us define $\{Q_M(t)\}_{t=1}^{\infty}$ as the key queue process with buffer size $M$,
and let $Q_M(1)=0$.
 Then, during each block $t$,
\begin{enumerate}
\item The transmitter and receiver agree on secret key bits of size $R_s(t)$ using
privacy amplification, and store the key on their secret key buffers.
\item The transmitter pulls key bits of size $\Rate$ from its secret key
 buffer to secure the message stream of size $\Rate$ using one time pad,
and transmits over the channel.
\end{enumerate}
as explained in Appendix~\ref{a:FullCSI}.
The last phase is skipped if outage (${\Oc}_{\text{enc}}(t)$) is declared, 
which is triggered by one of the following events
\begin{itemize}
\item Channel Outage (${\Oc}_{\text{ch}}(t)$): The channel cannot support reliable transmission at rate $\Rate$, i.e. $R_m(t)<\Rate$. 
\item Key Outage (${\Oc}_{\text{key}}(t)$): There are not enough key bits in the key queue to secure the message at rate $\Rate$. 
This event occurs when $Q_M(t)+ R_s(t) - \Rate < 0$.
\item Artificial outage (${\Oc}_{\text{a}}(t)$): Outage is artificially declared, even though reliable transmission at rate $\Rate$ is possible.
\end{itemize}
Due to the definition of $P^{\Rate}(\hv)$, $\Pr({\Oc}_{\text{ch}}(t))\leq \epsilon$ $\forall t$, and the set $\{{\Oc}_{\text{ch}}(t)\}$
of events indexed by $t$ are i.i.d. We choose $\{{\Oc}_{\text{a}}(t)\}$ such that
${\Oc}_{\text{x}}(t) = {\Oc}_{\text{ch}}(t) \cup {\Oc}_{\text{a}}(t)$ is i.i.d. as well, and
$$\Pr({\Oc}_{\text{x}}(t)) = \epsilon, ~ \forall t$$
The dynamics of the key queue can therefore be modeled by
\begin{align}
 Q_M(t+1)= \min ( M,Q_M(t)+R_s(t)- \ind(\bar{\Oc}_{\text{enc}}(t))\Rate ) \label{HatQkDynamics}
\end{align}  
Note that $Q_M(t)\geq 0$ $\forall t$, due to the definition of $\Oc_{\text{key}}(t)$.

Let $L^T(M)$ be the time average loss ratio over the first $T$ blocks, for buffer size $M$,
which is defined as the ratio of the amount of loss of key bits due to overflows, and the total amount of input key bits
\begin{align}
L^T(M)= \frac{\sum_{t=1}^T  \left( Q_M(t) + R_s(t) -
\ind(\bar{\Oc}_{\text{enc}}(t))\Rate-M \right)^+ }{\sum_{t=1}^T R_s(t)} \label{keyLoss}
\end{align}
Then, we can see that $\forall T>0$,
\begin{align}
(1-L^T(M))\sum_{t=1}^{T} R_s(t) = Q_M(T)+\sum_{t=1}^T \Rate \ind(\bar{\Oc}_{\text{enc}}(t)) \label{keydynamics}
\end{align}
follows from \eqref{HatQkDynamics}, \eqref{keyLoss}, and the fact that $Q_M(1)=0$.
\begin{lemma}\label{l:stationarity}
$Q_M(t)$ converges in distribution
to an almost surely finite random variable.
\end{lemma}
The proof is provided in Appendix~\ref{app:stationarity}. This implies that  
$\lim_{t\to \infty}\Pr(\Oc_{\text{enc}}(t))$  exists.
 Now, we provide our asymptotic result
for the key loss ratio.
We define the drift and variance of this process as
\begin{align}
\mu_{\Rate} &= \expect[R_s(\Hm,P^{\Rate})-\Rate \ind(\bar{\Oc}_{\text{x}}(t))] \nonumber\\ 
            &= \expect[R_s(\Hm,P^{\Rate})]-\Rate (1-\epsilon) \label{mu_rate}
\end{align} 
and    
\begin{align}     
\sigma_{\Rate}^2 & = \var [R_s(\Hm,P^{\Rate})-\Rate\ind(\bar{\Oc}_{\text{x}}(t))]  \nonumber    
\end{align}
respectively, where \eqref{mu_rate} follows from the definition of ${\Oc}_{\text{x}}(t)$.
\begin{lemma}\label{t:lossPr}
For any $M>0$, the key loss ratio satisfies the following asymptotic relationship
\begin{align}
\lim_{\Rate \searrow C_F^{\epsilon}} \lim_{T\to \infty} L^T\left(M \frac{{\sigma}_{\Rate}^2}{|{\mu}_{\Rate}|}\right)\frac{2|{\mu}_{\Rate}|\expect[R_s(\Hm,P^{\Rate})]  e^{\frac{-2\Rate|{\mu}_{\Rate}|}{\sigma_{\Rate}^2}} }{{\sigma}_{\Rate}^2}
 \leq e^{-2M}
\end{align}
\end{lemma} 
The proof is provided in Appendix~\ref{app:BufferOverflowProof}.
\begin{lemma}\label{outagelemma}
If $\lim_{t\to \infty}\Pr(\Oc_{\text{enc}}(t))=\epsilon'$, then
 $\epsilon'$ secrecy outage probability \eqref{outage} is satisfied. 
\end{lemma}
\begin{proof}
Find $B$ such that $\Pr(\Oc_{\text{enc}}(t))=\epsilon'+\delta$ for any $t>B$.
In 2-index time notation $(s,b)$ with $t= sB+b$,
it corresponds to $\Pr(\Oc_{\text{enc}}(s,b,\Rate))=\epsilon'+\delta$,
 $\forall(s,b):s\neq 1$.
Then.
\begin{align}
\Pr({\Oc}_{\text{sec}}(s,b,\Rate,\delta)) &\leq \Pr({\Oc}_{\text{sec}}(s,b,\Rate,\delta)|\bar{\Oc}_{\text{enc}}(s,b,\Rate))+\Pr({\Oc}_{\text{enc}}(s,b,\Rate)) \label{corr1:oeq}\\
                              &\leq \Pr({\Oc}_{\text{enc}}(s,b,\Rate))  \label{corr1:oenc}\\
                              &  \leq \epsilon'+\delta \label{corr1:eq3}
\end{align}
Here, \eqref{corr1:oeq} follows from the union bound, and second term follows 
from the equivocation analysis \eqref{appa:oeq1} and \eqref{appa:oeq2}
in Appendix A, which shows that there exists some packet size $N$ large enough such that $\Pr({\Oc}_{\text{sec}}(s,b,\Rate,\delta)|\bar{\Oc}_{\text{enc}}(s,b,\Rate))= 0$.
Equation \eqref{corr1:eq3} implies that $\epsilon'$ secrecy outage probability \eqref{outage}
is satisfied.
\end{proof}
Let $\lim_{t\to \infty}\Pr(\Oc_{\text{enc}}(t))=\epsilon'$.
Since $\Pr(\Oc_{\text{x}}(t)) = \epsilon$ and $\Oc_{\text{enc}}(t) = \Oc_{\text{x}}(t)
\cup \Oc_{\text{key}}(t)$, we have $\lim_{t\to\infty}\Pr(\Oc_{\text{key}}(t))>0$.
This implies  that $\lim_{T \to \infty}\frac{1}{T}Q_{M}(T) =0$ 
(since otherwise, key outage probability would be zero),
which, due to \eqref{keydynamics} implies
\begin{align}
(1-\lim_{T \to \infty} L^T(M) )\expect[R_s(\Hm,P^{\Rate})] 
&=  (1- \lim_{t \to \infty}\Pr({\Oc}_{\text{enc}}(t)))\Rate  \nonumber \\  
&=  (1-\epsilon')\Rate \label{corr1:ltm}
\end{align}
Here, due to the choice of power allocation function $P^{R}(\hv)$,
we have $\expect[R_s(\Hm,P^{\Rate})] = \lim_{T\to\infty}\frac{1}{T}\sum_{t=1}^T R_s(t)$.
Plugging the result of Lemma~\ref{t:lossPr} into \eqref{corr1:ltm}, we obtain 
the required key buffer size to achieve $\epsilon'$ probability of secrecy outage
\begin{align}
\lim_{\Rate \searrow C_F^{\epsilon}} \frac{M_{\Rate}(\epsilon')-R}{\frac{{\sigma}_{\Rate}^2}{2|{\mu}_{\Rate}|}\log \left(  \frac{{\sigma}_{\Rate}^2} 
  {2|{\mu}_{\Rate}|\left( \expect[R_s(\Hm,P^{\Rate})] - (1-\epsilon')\Rate \right)  }  \right)} \leq 1  \label{Mbeta}
\end{align}
 We know from  \eqref{SecrecyCapacity} that $\epsilon$ and $\epsilon'$-achievable secrecy capacities satisfy the conditions
$C_F^{\epsilon'}(1-\epsilon')= \expect[R_s(\Hm,P^{\Rate})] |_{\Rate=C_F^{\epsilon'}}$ and 
$C_F^{\epsilon}(1-\epsilon)=\expect[R_s(\Hm,P^{\Rate})] |_{{\Rate}=C_F^{\epsilon}}=\expect[R_s(\Hm,P^{\ast})]$, respectively.
By Lemma~\ref{l:Rb}, we know that $\expect[R_s(\Hm,P^{\Rate})]$ is a continuous function of $\Rate$, hence for any given $\epsilon'>\epsilon$, there exists an ${\Rate}$ such that $C_F^{\epsilon}<{\Rate}<C_F^{\epsilon'}$,
and  $\expect[R_s(\hv,P^{\Rate})]=(1- \frac{\epsilon+\epsilon'}{2}){\Rate}$.
Furthermore,
as $\epsilon' \to \epsilon$, $C_F^{\epsilon'} \to C_F^{\epsilon}$.
Let us define a monotonically decreasing sequence $(\epsilon'(1), \epsilon'(2), \cdots)$, such that $\lim_{i \to \infty}\epsilon'(i)=\epsilon$.
For any $i \in \nat$, find ${\Rate}(i)$ such that $C_F^{\epsilon}<{\Rate}(i)<C_F^{\epsilon'(i)}$, 
and $\expect[R_s(\hv,P^{{\Rate}(i)})]=(1- \frac{\epsilon+\epsilon'(i)}{2})\Rate(i)$, therefore $\mu_{{\Rate}(i)} = (\epsilon - \epsilon')/(2{\Rate}(i))$.
From \eqref{Mbeta}, we get
\begin{align*}
\lim_{i \to \infty} \frac{M_{{\Rate}(i)}(\epsilon'(i))-\Rate(i)}{\frac{\sigma_{{\Rate}(i)}^2}{(\epsilon'-\epsilon){\Rate}(i)}\log \left( \frac{{\sigma}_{{\Rate}(i)}^2}{{\Rate}(i)
(\epsilon'(i)-\epsilon)}\right)} \leq 1
\end{align*}
Since as $i\to \infty$, ${\Rate}(i)\to C_F^{\epsilon}$, $\epsilon'(i)\to \epsilon$ and 
$\sigma_{{\Rate}(i)}^2 \to \sigma_{C_F^{\epsilon}}^2$, where
\begin{align*}
\sigma_{C_F^{\epsilon}}^2 &=    \var[R_s(\Hm,P^{C_F^{\epsilon}}) - C_F^{\epsilon}\ind(\bar{\Oc}_{\text{x}}(t))] \\
                          &\leq \var[R_s(\Hm,P^{C_F^{\epsilon}})] - C_F^{\epsilon}(1-\epsilon)\epsilon
\end{align*}
The last inequality induces the upper bound \eqref{c:asymptoticbuffer}, which concludes the proof.
\end{proof}

\section{Numerical Results}\label{section:simulations}
In this section, we conduct simulations to illustrate our main results with two examples.
In the first example, we analyze the relationship between $\epsilon$-achievable
secrecy capacity and average power. We assume that both the main channel and eavesdropper channel are characterized
by Rayleigh fading, where the main channel and eavesdropper channel power gains follow exponential distribution
with means 2 and 1, respectively. 
Since Rayleigh channel is non-invertible, maintaining a non-zero secrecy rate with zero secrecy outage probability is
impossible. In Figure~\ref{fig:capacityresults},
we plot the $\epsilon$-achievable secrecy capacity
as a function of the average power, for $\epsilon = 0.02$ outage probability, for
both full CSI and main CSI cases. It can be clearly observed from the figure that the gap between 
capacities under
full CSI and main CSI vanishes as average power increases, which support the result of Theorem~\ref{t:highpowercapacity}.
\noindent \begin{figure}[h]
\centerline{\includegraphics[width=8.8cm]{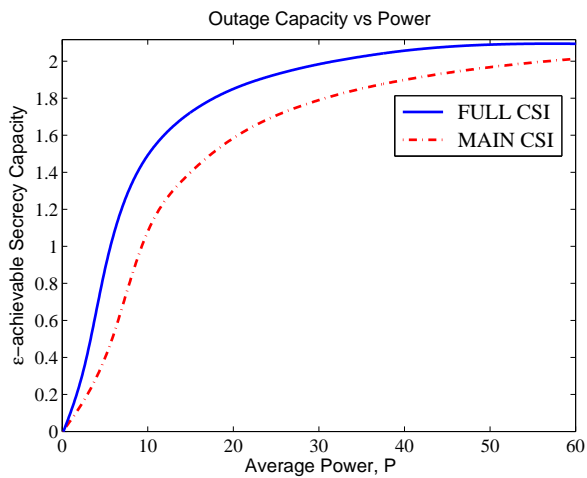}}
    \caption{The $\epsilon$-achievable secrecy capacities as a function of average power, $P_{\text{avg}}$}
    \label{fig:capacityresults}
\end{figure}

In the second example, we study the relationship between the buffer size, key loss ratio and the outage probability. 
We assume that both the main and eavesdropper channel gains
follow a chi-square distribution of degree 2, but with means 2 and 1, respectively.  
We focus on the full CSI case, and consider the scheme described in Section~\ref{s:finitebuffer}.
We consider transmission at secrecy rate of $\Rate$ with the use of the power allocation policy $P^{\Rate}(\hv)$ that solves
the problem \eqref{eq:PowerControl}-\eqref{serviceconstraint}. 
For $\epsilon =0.02$,
and the average power $P_{\text{avg}}=1$,
we plot the key loss ratio  \eqref{keyLoss}, 
as a function of buffer size $M$ in Figure~\ref{fig:buffersizing1}, for 
$\Rate= C_F^{\epsilon}$, $\Rate=1.01 C_F^{\epsilon}$ 
and $\Rate=1.02 C_F^{\epsilon}$, where $C_F^{\epsilon}$ is the $\epsilon$-achievable secrecy capacity.
It is shown in Lemma~\ref{t:lossPr} of Section~\ref{s:finitebuffer} that 
expect the key loss ratio $L^T(M)$
 decreases as ${\Rate}$ increases, which is observed in Figure~\ref{fig:buffersizing1}. 
Finally, we study the relationship between the secrecy outage probability and the buffer size
for a given rate. 
In Figure~\ref{fig:buffersizing2}, we plot the secrecy outage probabilities, denoted as $\epsilon'$, 
as a function of buffer size $M$ for the same encoder parameters.
On the same graph, we also plot our asymptotic result given in Theorem~\ref{eq:BuffervsOutage},
which provides an upper bound on the required buffer size to achieve $\epsilon'$ outage probability for rate $C_F^{\epsilon}$, 
with the assumption that \eqref{c:asymptoticbuffer} is an 
equality for any $\epsilon'$.  We can see that, this theoretical result serves as an upper bound on
 the required buffer size when $\epsilon'-\epsilon$, which is the additional secrecy outages due to key buffer overflows,
 is very small. Another important observation from Figures~\ref{fig:buffersizing1} and \ref{fig:buffersizing2} is that,
 for a fixed buffer size, although the key loss ratio decreases as $\Rate$ increases,
 secrecy outage probability increases. This is due to the fact that key bits are pulled from the key queue at a faster rate,
 hence the decrease in the key loss ratio does not compensate for the increase of the rate that key bits are 
pulled from the key queue, therefore the required buffer size to achieve same $\epsilon'$
is higher for larger values of $\Rate$.
\begin{figure}[ht]
\begin{minipage}[b]{0.45\linewidth}
\centering
\includegraphics[scale=0.40]{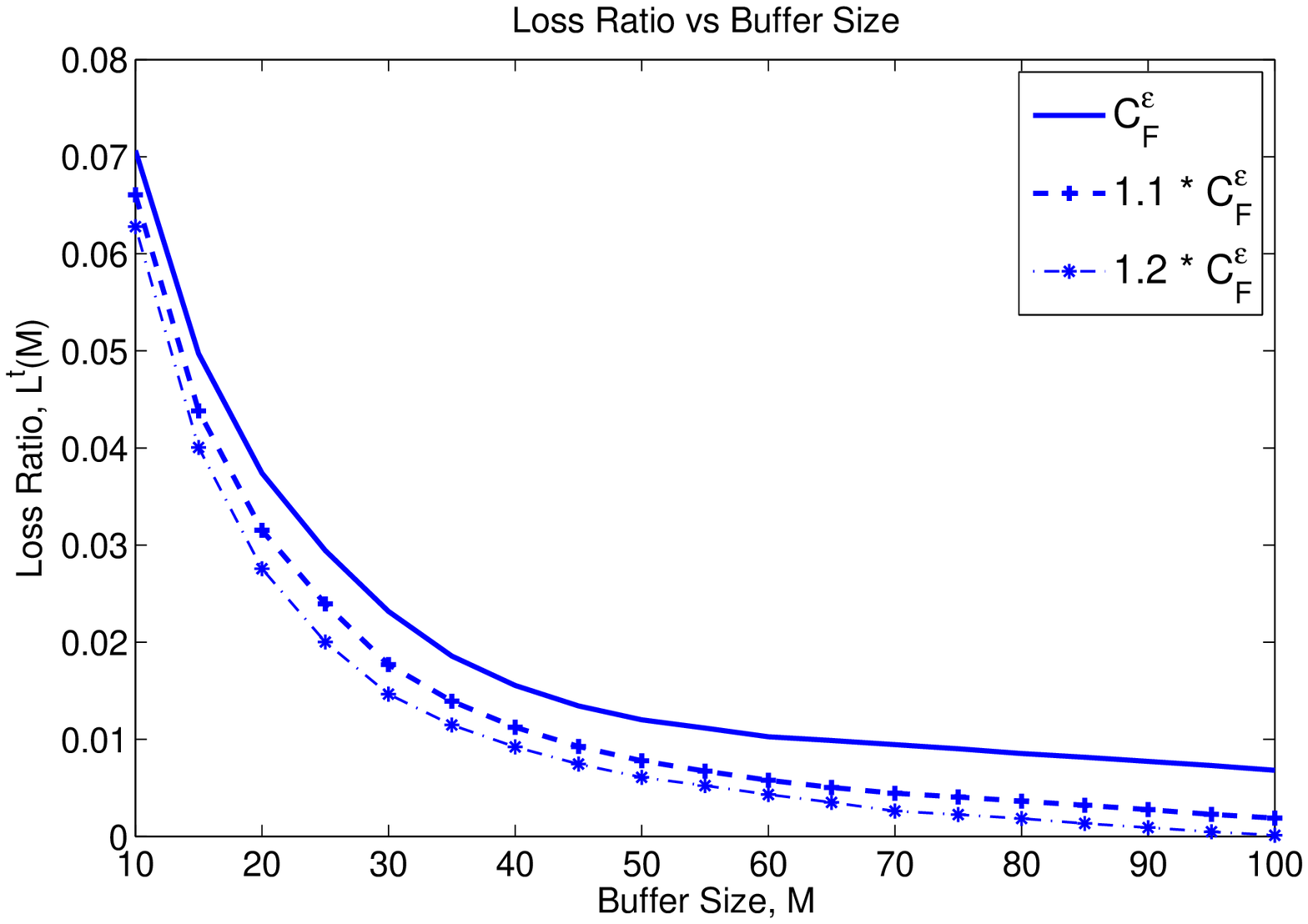}
\caption{Relationship between buffer size $M$, and key loss ratio $L^t(M)$ }
\label{fig:buffersizing1}
\end{minipage}
\hspace{0.5cm}
\begin{minipage}[b]{0.45\linewidth}
\centering
\includegraphics[scale=0.40]{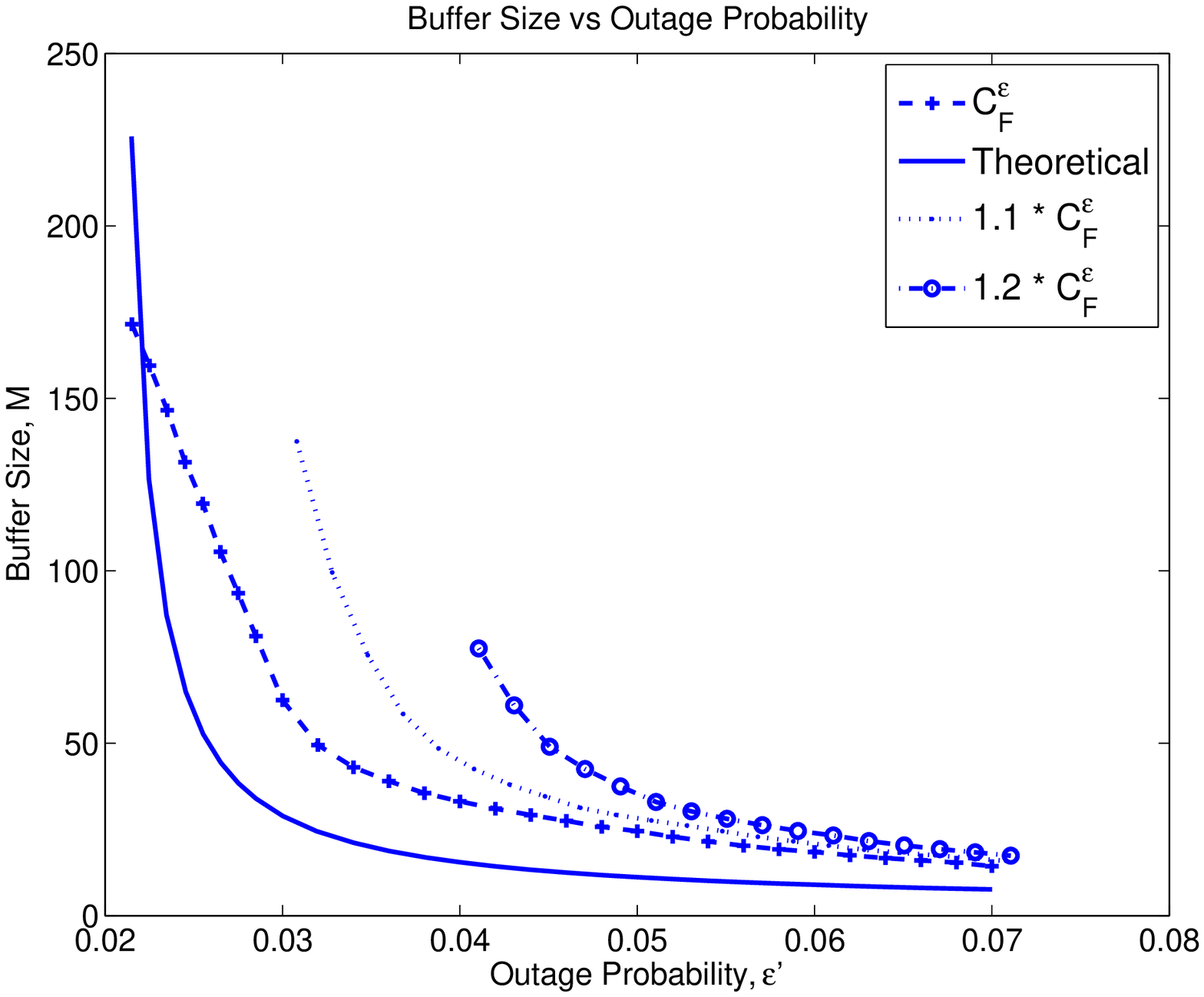}
\caption{Relationship between buffer size $M$, and outage probability $\epsilon'$}
\label{fig:buffersizing2}
\end{minipage}
\end{figure}

\section{Conclusions}\label{s:conclusion}
This paper obtained sharp characterizations of the secrecy outage capacity of block flat fading channels
under the assumption full and main CSI at the transmitter. In the two cases, our achievability scheme 
relies on opportunistically exchanging private keys between the legitimate nodes and using them later to 
secure the delay sensitive information. We further derive the optimal power control policy in each scenario 
revealing an interesting structure based by judicious time sharing between time sharing and the optimal strategy 
for the ergodic. Finally, we investigate the effect of key buffer overflow on the secrecy outage probability
when the key buffer size is finite.

\begin{appendices}

\section{Proof of Theorem~\ref{FullCSICapacity}}\label{a:FullCSI}
First, we prove the achievability.
Consider a fixed power allocation function $P(\hv) \in \Pc'$.
Let us fix ${\Rate}< \expect[R_s(\Hm,P)]/(1-\epsilon)$. We show that for any $\delta>0$, there exist some $B$ and $N$ 
large enough such that
the constraints in \eqref{errorconstraint} and \eqref{outage} are satisfied, which implies that any ${\Rate}<\expect[R_s(\Hm,P)]/(1-\epsilon)$ is 
an $\epsilon$-achievable secrecy rate. The outage capacity is then 
found by maximizing $\expect[R_s(\Hm,P)]/(1-\epsilon)$
over the set $\Pc'$ of power allocation functions.

Our scheme utilizes secret key buffers at both the transmitter
and legitimate receiver.
Then, \\
i) At the end of every block $(s,b)$, using privacy amplification, 
legitimate nodes (transmitter and receiver) generate $N(R_s(s,b)- \delta)$ bits of secret key from
the transmitted signal in that particular block, and store it in their secret key buffers. We denote the
generated secret key at the transmitter as $V(s,b)$, and at the receiver as $\hat{V}(s,b)$.\\
ii) At every block $(s,b)$, $s\neq 1$, the transmitter pulls $N{\Rate}$ bits from its secret key buffer to secure the outage
constrained message of size $H(W(s,b))=N{\Rate}$, using Vernam's one time pad. The receiver uses the same key to correctly decode
the message. We denote the pulled key at the transmitter as $K(s,b)$, and at the receiver as $\hat{K}(s,b)$.
For simplicity in analysis, we assume that keys generated at $s-1$'th superblock
are used only in the $s$'th superblock.
This stage is skipped in the first super-block, and when 'encoder' outage
${\Oc}_{\text{enc}}(s,b,{\Rate})$ occurs, which is the union of the following events:
\begin{itemize}
\item Channel outage (${\Oc}_{\text{ch}}(s,b,{\Rate})$): Channel is not suitable for reliable transmission at rate ${\Rate}$, i.e., 
$R_m(s,b)<{\Rate}$.
\item Key outage (${\Oc}_{\text{key}}(s,b,{\Rate})$): There are not enough key bits in the key queue to secure $W(s,b)$, i.e.,
\begin{align*}
\left(\sum_{b'=1}^B H(V(s-1,b')) -\sum_{b'=1}^b H(K(s,b'))\right) < 0 
\end{align*}
\item Artificial outage (${\Oc}_{\text{a}}(s,b,{\Rate})$): The transmitter declares 'outage', even though reliable secure
transmission of $W(s,b)$ is possible. This is introduced to bound the probability of key outages, and
is explained in the outage analysis.
\end{itemize}

\textbf{Encoding:}\\
Our random coding arguments rely on an ensemble of codebooks
generated according to a zero mean Gaussian distribution with variance $P(s,b)$ \footnote{Note that, it is
also possible
to use a finite number of codebooks by partitioning the set  $\{\hv\}$ of channel gains, and
using a different Gaussian codebook for every partition \cite{Secrecy:08}.
}.\\
1) When ${\Oc}_{\text{enc}}(s,b,{\Rate})$ does not occur, the message is secured with the secret key bits pulled from the key queue,
using one time pad\footnote{We assume that both the message and the key are converted to binary form in this process.} 
\begin{align}
W_{sec}(s,b) = W(s,b) \oplus K(s,b)
\end{align}
Clearly, $W_{sec}(s,b) \in \Wc_{sec}=\{1,\cdots,2^{N{\Rate}} \}$. Furthermore, let 
$\Wc_{x1}(s,b) = \{ 1,\cdots,2^{N(R_m(s,b)-{\Rate}-\delta)} \}$. 
To transmit the one time padded message $w_{sec}(s,b)$, the encoder randomly and uniformly
chooses $w_{x1}(s,b)$ among $\Wc_{x1}(s,b)$, and transmits to codeword $\xv(s,b)$ 
indexed by $(w_{sec}(s,b),w_{x1}(s,b))$ over the channel. \\
2) When ${\Oc}_{\text{enc}}(s,b,{\Rate})$ occurs, $W(s,b)$ is not transmitted. Let $\Wc_{x2} = \{ 1,\cdots,2^{N(R_m(s,b)-\delta)} \}$.
The encoder randomly and uniformly
chooses $w_{x2}(s,b)$ among $\Wc_{x2}$, and transmits to codeword $\xv(s,b)$ indexed by $w_{x2}(s,b)$ over the channel.\\
The reason for transmitting $w_{x1}(s,b)$ and $w_{x2}(s,b)$ is to confuse 
the eavesdropper to the fullest extent in the privacy amplification process.\\
\textbf{Decoding:}\\
1) When ${\Oc}_{\text{enc}}(s,b,\Rate)$ does not occur, the receiver finds the jointly 
typical $(\hat{w}_{sec}(s,b),\hat{w}_{x1}(s,b),\yv(s,b))$
pair, where $\yv(s,b)$ denotes the received signal at block $(s,b)$. Then, using one-time pad, the receiver
obtains $\hat{w}(s,b) = \hat{w}_{sec}(s,b) \oplus \hat{k}(s,b)$. \\
2) When ${\Oc}_{\text{enc}}(s,b,\Rate)$ occurs, the receiver finds the jointly 
typical $(\hat{w}_{x2}(s,b),\yv(s,b))$.

Define the error events 
\begin{align*}
E_1(s,b) &= \left\{(\hat{w}_{sec}(s,b),\hat{w}_{x1}(s,b)) \neq ({w}_{sec}(s,b),{w}_{x1}(s,b)) \right\} \\
E_2(s,b) &= \left\{\hat{w}_{x2}(s,b) \neq {w}_{x2}(s,b) \right\}  \\
E_3(s,b,\delta) &= \left\{\frac{1}{N}\|X(s,b)\|^2 > P(s,b)+\delta\right\}
\end{align*}
Independent of whether the event ${\Oc}_{\text{enc}}(s,b,\Rate)$ occurs or not,
the encoding rate is equal to $R_m(s,b)-\delta$, which is below the supremum of achievable
main channel rates. Furthermore, each element of $\Xm(s,b)$
is independently drawn from Gaussian distribution of mean $0$
and variance $P(s,b)$. Therefore, random coding arguments guarantee us that $\forall B>0$,
 $\exists N_1>0$ such that $\forall N>N_1$,
$\Pr(E_1(s,b))\leq \frac{\delta}{3B}$, $\Pr(E_2(s,b))\leq \frac{\delta}{3B}$ and $\Pr(E_3(s,b))\leq \frac{\delta}{3}$. 

\textbf{Privacy Amplification:} At the end of every block $(s,b)$, the transmitter and receiver generate
secret key bits, by applying a universal hash function on the exchanged signals in that particular block.
First, we provide the definition of a universal hash function.
\begin{definition} (\cite{Bennett:95})
A class $G$ of functions $\Ac  \to \Bc$ is universal, if for any  $x_1 \neq x_2$ in
 $\Ac$, the probability that $g(x_1)=g(x_2)$
is at most $\frac{1}{\Bc}$ when $g$ is chosen at random from $G$ according to a uniform distribution.
\end{definition}
\begin{lemma}\label{l:privacyamp}
For any $B>0$, there exists $N_2(B)>0$ such that, $\forall N>N_2(B)$, and for any block $(s,b)$,
\begin{itemize}
\item When ${\Oc}_{\text{enc}}(s,b,\Rate)$ does not occur, the transmitter and receiver can generate 
secret key bits $V(s,b) = G([W_{sec}(s,b)~ W_{x1}(s,b)])$ and $\hat{V}(s,b) = G([\hat{W}_{sec}(s,b)~ \hat{W}_{x1}(s,b)])$
respectively, such that $V(s,b) = \hat{V}(s,b)$ if the error event $E_1(s,b)$ does not occur, and
\begin{align} 
H(V(s,b)) &= N(R_s(s,b)-\delta) \label{privacyamp1}\\
\frac{1}{N}I(V(s,b);\Zm^{SB},\hv^{SB},G) &\leq \delta/B \label{privacyamp2} 
\end{align}
\item Similarly, when ${\Oc}_{\text{enc}}(s,b,\Rate)$  occurs,
the transmitter and receiver can generate 
secret key bits $V(s,b) = G([W_{x2}(s,b)])$ and $\hat{V}(s,b) = G([\hat{W}_{x2}(s,b)])$
respectively, such that $V(s,b) = \hat{V}(s,b)$ if the error event $E_2(s,b)$ does not occur, and
\eqref{privacyamp1}, \eqref{privacyamp2} are satisfied.
\end{itemize}
\end{lemma}
\begin{proof}
The proof follows follows the approach of \cite{Nascimento}, which applies privacy amplification to 
Gaussian channels. 
First, we introduce the information theoretic quantities required for the proof. For random variables $A,B$, define
\begin{itemize}
\item Renyi entropy of $A$ as $\log\expect[P_{A}(a)]$
\item Min-entropy as $A$ as $H_{\infty}(A) = \min_{a}\log\left(\frac{1}{\Pr_{A}(a)} \right)$.
\item Conditional min-entropy of $A$ given $B$ as $H_{\infty}(A|B) = \inf_{b}H_{\infty}(A|B=b).$
\item $\delta$-smooth min-entropy of $A$ as $H_{\infty}^{\delta}(A) = \max_{A':\|\Pr_{A}-\Pr_{A'}\|<\delta}H_{\infty}(A')$.
\end{itemize}

Without loss of generality, we drop the block index $(s,b)$ and $\Rate$, and focus on the first block $(1,1)$, and assume 
the event ${\Oc}_{\text{enc}}$ does not occur.  Let $W_{X} = [W_{sec}~W_{x1}]$, with sample realization 
sequences denoted by $w_x$. 
Let $V = G(W_X)$, where $G$ denotes a random universal hash function that maps $W_X$
 to to an r-bit binary message $V \in \{0,1\}^r$. Then, it is clear that if error event $E_1$ does not occur,
$\hat{V}=V$ since $W_X = \hat{W}_X$, for any choice of $G$. To show that the security constraints \eqref{privacyamp1}-\eqref{privacyamp2}
are satisfied, we  
cite the privacy amplification theorem, which is originally defined for 
discrete channels. For this purpose, we 
define a quantization function $\phi$, with sensitivity parameter $\Delta = \sup_{\zv}|\zv-\phi(\zv)|$. 
Let $\Zm^{\Delta} = \phi(\Zm)$ denote the
 quantized version of $\Zm$. 
where $\zv^{\Delta}$ denotes realization sequences. Then, by Theorem 3 of \cite{Bennett:95} there exists a 
universal function $G$ such that \footnote{We omit $\hv^{SB}$ in the following parts of the proof of Lemma~\ref{l:privacyamp} 
for notational simplicity.}
\begin{align*}
H(G(W_{X}) | \Zm^{\Delta}=\zv^{\Delta}, G) \geq r - \frac{2^{r-R(W_{X}|\Zm^{\Delta}=\zv^{\Delta})}}{\ln 2}
\end{align*}
 Now, we relate this expression to the Shannon entropy of the message, conditioned on eavesdropper's
actual received signal.
Using the facts $H_{\infty}(W_X) \leq R(W_X)$ and $H_{\infty}(W_X|\Zm^{\Delta},G) \leq H_{\infty}(W_X|\Zm^{\Delta}=\zv^{\Delta},G)$, 
it is easy to show that
\begin{align*}
H(G(W_{X}) | \Zm^{\Delta}, G) \geq r - \frac{2^{r-H_{\infty}(W_{X}|\Zm^{\Delta})}}{\ln 2} 
\end{align*}
Then, due to the asymptotic relationship between continuous random variables and their quantized versions \cite{TCover},
there exists a quantization function $\phi$ such that
$\Delta$ is small enough, and
\begin{align}
H(G(W_X)|G,\Zm) &\geq H(G(W_X)|G,\Zm^{\Delta}) -\frac{\delta}{2B}  \nonumber\\
                &\geq r - \frac{2^{r - H_{\infty}(W_X|\Zm^{\Delta})}}{\ln 2} -\frac{\delta}{2B}  \label{appa:privampeq5}
\end{align}
are satisfied. To relate min-entropy to Shannon entropy, we use the result of
 Theorem 1 of \cite{Nascimento}; 
$\forall \delta'>0$, $\exists$ a block length $N'$ such that $\forall N>N'$,
\begin{align}
\frac{1}{N} H(\Xm^{\Delta}|\Zm^{\Delta}) \leq \frac{1}{N} H_{\infty}^{\delta'}(\Xm^{\Delta}|\Zm^{\Delta}) + \delta/B
 \label{appa:privampeq4}
\end{align}
Now, we proceed as follows,
\begin{align}
H_{\infty}(W_X|\Zm^{\Delta}) &= \lim_{\delta' \to 0} H_{\infty}^{\delta'}(W|\Zm^{\Delta}) \nonumber\\
                             &\geq H(W_X) - I(W_X;\Zm^{\Delta}) - N\delta/B \label{appa:privampeq3}\\
                             &\geq H(W_X) - I(\Xm;\Zm) - N\delta/B \label{appa:privampeq1}\\
                             &=   N R_s-N \delta/B  \label{appa:privampeq2}
\end{align}
where \eqref{appa:privampeq3} follows from \eqref{appa:privampeq4}, and the appropriate choice of $N'$.
 \eqref{appa:privampeq1} follows from the fact that $W_X \rightarrow \Xm \rightarrow \Zm \rightarrow \Zm^{\Delta}$ forms
a Markov chain.
 \eqref{appa:privampeq2} follows
from the fact that $H(W_X) = N(R_m - \delta)$, and similarly $I(\Xm;\Zm) \leq N(R_m - R_{s} - \delta)$, which
is the eavesdropper's maximum achievable rate.
 For the choice of $H(V) = r = N(R_s - \delta)$, from \eqref{appa:privampeq5}, \eqref{appa:privampeq2}, 
 and the fact that $V=G(W_X)$, we get
\begin{align}
I(V;G,\Zm) &=  H(G(W_X)) - H(G(W_X)|\Zm,G) \nonumber\\
           &\leq \frac{2^{-N(B-1)/B}}{\ln 2} + \frac{\delta}{2B} \nonumber\\
           &\leq \frac{\delta}{B} \label{appa:privampeq6}
\end{align} 
since there exists some $N''$ such that forall $N> N''$, \eqref{appa:privampeq6}
Hence, for $N> N_2 = \max(N',N'')$, the constraints \eqref{privacyamp1}, \eqref{privacyamp2} are satisfied.
The proof for the case where ${\Oc}_{\text{enc}}$ occurs is very similar, and is omitted.
\end{proof}

\textbf{Equivocation Analysis:} 
Secrecy outage probability can be bounded above as
\begin{align}
\Pr({\Oc}_{\text{sec}}(s,b,\Rate,\delta)) &= \Pr({\Oc}_{\text{eq}}(s,b,\Rate,\delta)\cup {\Oc}_{\text{ch}}(s,b,\Rate)) \nonumber\\ 
                          &\leq  \Pr({\Oc}_{\text{eq}}(s,b,\Rate,\delta)\cup {\Oc}_{\text{enc}}(s,b,\Rate))\nonumber\\
&\leq \Pr({\Oc}_{\text{eq}}(s,b,\Rate,\delta)|\bar{\Oc}_{\text{enc}}(s,b,\Rate)) + \Pr({\Oc}_{\text{enc}}(s,b,\Rate)) \label{appa:eq1}
\end{align}
where the first equality follows from the definition of secrecy outage \eqref{eq:secrecyoutage}, and \eqref{appa:eq1} follows
from the union bound.
Now, we upper bound the first term.
For the choice $N,B$ such that  $N= \max(N_1(B),N_2(B))$,
the equivocation at every block $(s,b)$ in case of no encoder outage can be bounded as
\begin{align}
H&(W(s,b)| \Zm^{SB},\hv^{SB},G,\bar{\Oc}_{\text{enc}}(s,b,\Rate)) \nonumber\\
&\stackrel{(a)}{=} H(W(s,b)|\Zm(s-1,:),\hv^{SB},G,\bar{\Oc}_{\text{enc}}) \nonumber\\
- &I(W(s,b);\Zm(s,b)|\Zm(s-1,:),\hv^{SB},G,\bar{\Oc}_{\text{enc}}) \nonumber\\
                                 &\stackrel{(b)}{\geq} H(W(s,b)) - I\big{(}W(s,b);\Zm(s,b),W_{sec}(s,b)|\nonumber\\
&\Zm(s-1,:),\hv^{SB},G,\bar{\Oc}_{\text{enc}})\big{)} \nonumber\\
                                 &\stackrel{(c)}{\geq} N{\Rate} - I(W(s,b);W_{sec}(s,b)|\Zm(s-1,:),\hv^{SB},G,\bar{\Oc}_{\text{enc}})) \nonumber\\
								 &\stackrel{(d)}{\geq} N{\Rate} - H(K(s,b)|\Zm(s-1,:),\hv^{SB},G,\bar{\Oc}_{\text{enc}})) \nonumber\\
								 &\stackrel{(e)} {\geq} N{\Rate} - \sum_{b=1}^B H(V(s-1,b)|\Zm(s-1,:),\hv^{SB},G,\bar{\Oc}_{\text{enc}})) \nonumber\\
								 &\geq N({\Rate} -\delta) \label{appa:oeq1}
\end{align}
where we use the notation $\Zm(s,:) = \{\Zm(s,i)\}_{i=1}^B$,
and omit the index $(s,b,\Rate)$ from $\bar{\Oc}_{\text{enc}}(s,b,\Rate)$.
Notice that $W(s,b) = W_{sec}(s,b) \oplus K(s,b)$, and due to our encoder structure, $W_{sec}(s,b)$
is transmitted only in $(s,b)$'th block, and similarly, $K(s,b)$ is generated in $s-1$'th superblock.
Then, due to the memoryless property of the channel, $W(s,b)\rightarrow (\Zm(s,b),\Zm(s-1,:))\rightarrow \Zm^{SB}$, hence $(a)$ follows.
The first term of $(b)$ follows since $W_{sec}(s,b)$ is independent of $\Zm(s-1,:)$, and the second term of $(b)$ follows since
$W(s,b)\rightarrow W_{sec}(s,b)\rightarrow \Zm(s,b)$ forms a Markov chain.
$(c)$ follows due to $H(W(s,b))=N\Rate$.
$(d)$ follows since there is no encoder outage, hence key outage, and $(e)$ follows since $K(s,b)$ is
pulled from the key buffers, which contain the pool of key bits $\{V(s-1,b)\}_{b=1}^B$ generated during superblock $s-1$, 
and \eqref{appa:oeq1} follows from \eqref{privacyamp2}. 
Then,
\begin{align}
\Pr({\Oc}_{\text{eq}}(s,b,\Rate,\delta)|\bar{\Oc}_{\text{enc}}(s,b,\Rate))&= \Pr(\frac{1}{N}H(W(s,b)| \Zm^{SB},\hv^{SB},G)<\Rate-\delta) \nonumber\\
                         &= 0 \label{appa:oeq2}
\end{align}
Now, we bound the encoder outage probability. By the union bound,
\begin{align*}
\Pr({\Oc}_{\text{enc}}(s,b,\Rate))&\leq \Pr({\Oc}_{\text{ch}}(s,b,\Rate) \cup {\Oc}_{\text{a}}(s,b,\Rate))+ \Pr({\Oc}_{\text{key}}(s,b,\Rate))                     
\end{align*}
Since $P(\hv) \in \Pc'$, due to the definition in \eqref{eq:OutageConstraint}, \eqref{eq:PowerConstraint2},
$\forall (s,b)$
\begin{align*}
\Pr({\Oc}_{\text{ch}}(s,b,\Rate)) &= \Pr(R_m(s,b)<\Rate) \\
                    &\stackrel{(f)}{=} \Pr\left(R_m(\Hm,P)<\frac{\expect[R_s(\Hm,P)]}{1-\epsilon}\right) \leq \epsilon
\end{align*}
where in $(f)$, we interchangeably use $R_m(s,b)\equiv R_m(\hv,P)$ due to stationarity of $P(\hv)$.
Note that, the events ${\Oc}_{\text{ch}}(s,b,\Rate)$ indexed by $(s,b)$ are i.i.d. Here, we introduce i.i.d. artificial outages
 ${\Oc}_{\text{a}}(s,b,\Rate)$ such that
 $$\Pr({\Oc}_{\text{ch}}(s,b,\Rate)\cup {\Oc}_{\text{a}}(s,b,\Rate)) =
 \epsilon, ~ \forall (s,b) $$
This would help us bound the probability of key outage.
For $(s,b)$, $s\neq 1$
\begin{align}
\Pr({\Oc}_{\text{key}}(s,b,\Rate))& = \Pr\left(\sum_{i=1}^B H(V(s-1,i)) - \sum_{i=1}^b H(K(s,i)) <0 \right) \nonumber\\
                   & = \Pr \bigg{(} \sum_{i=1}^B N(R_s(s-1,i)-\delta)- 
\sum_{i=1}^b N\Rate\ind(\bar{\Oc}_{\text{ch}}(s,i,\Rate)\cap \bar{\Oc}_{\text{a}}(s,i,\Rate)) <0 \bigg{)} \nonumber\\
                   & \leq \Pr\left(\sum_{i=1}^B \left[R_s(s-1,i)-\delta - \Rate\ind(\bar{\Oc}_{\text{ch}}(s,i,\Rate) \cap \bar{\Oc}_{\text{a}}(s,i,\Rate))\right] < 0\right) \label{appa:keyqueue}
\end{align}
Note that, the expression in \eqref{appa:keyqueue} represents a random walk with expected drift 
$\mu = \expect[R_s(\Hm,P)]-\delta - \Rate (1-\epsilon)$. For $\Rate\leq\frac{\expect[R_s(\Hm,P)]-\delta}{1-\epsilon}$, $\mu >0$, hence by 
the law of large numbers,
$\exists B_1>0$ such that $\forall B>B_1$, $\Pr({\Oc}_{\text{key}}(s,b,\Rate))< \delta$, $s \neq 1$.
Therefore, due to union bound and \eqref{appa:eq1}, \eqref{appa:oeq1}, \eqref{appa:oeq2}, for the choice $B=B_1$, $N = \max(N_1(B_1),N_2(B_1))$,
$\Pr({\Oc}_{\text{sec}}(s,b,\Rate,\delta))\leq \epsilon + \delta$,
which satisfies \eqref{outage}.

\textbf{Error Analysis:}
For $N,B$ such that $B= B_1$, $N = \max(N_1(B_1),N_2(B_1))$, $\forall (s,b)$, $s\neq 1$,
\begin{align*}
 \Pr(E(s,b,\delta)|\bar{\Oc}_{\text{enc}}(s,b,\Rate)) & \leq \Pr(W(s,b) \neq \hat{W}(s,b)) + 
 \Pr\left( \frac{1}{N} \|\Xm(s,b)\|^2 > P(s,b)+ \delta \right) \\
							      &= \Pr(W_{sec}(s,b)\oplus K(s,b) \neq \hat{W}_{sec}(s,b)\oplus \hat{K}(s,b)) + 
\Pr\left(\frac{1}{N} \| \Xm(s,b)\|^2 > P(s,b)+ \delta \right)  \\
							      &\leq \Pr(W_{sec}(s,b) \neq \hat{W}_{sec}(s,b)) + \Pr(K(s,b)\neq \hat{K}(s,b))   + \Pr\left(\frac{1}{N} \| \Xm(s,b)\|^2 > P(s,b)+ \delta \right)
\end{align*}
where 
 the first term can be bounded as
 $\Pr(W_{sec}(s,b) \neq \hat{W}_{sec}(s,b))\leq \frac{\delta}{3B}$ due to definition of $E_1(s,b)$,
and the choice of $N$. 
Similarly, the third term  can be bounded as 
$\Pr(W_{sec}(s,b) \neq \hat{W}_{sec}(s,b))\leq \delta/3$ due to definition of $E_3(s,b)$,
and the choice of $N$. 
The second term can be bounded as
\begin{align*}
\Pr(K(s,b)\neq \hat{K}(s,b)) &\stackrel{(a)}{\leq} 1- \prod_{i=1}^{B}\Pr(V(s-1,i) = \hat{V}(s-1,i)) \\
                             &\leq \sum_{i=1}^{B}\Pr(V(s-1,i)\neq \hat{V}(s-1,i)) \\
                             &\leq\sum_{i=1}^{B}\big{(}\Pr(E_{1}(s-1,i))\Pr(\bar{\Oc}_{\text{enc}}(s-1,i,\Rate))\\
                             + & \Pr(E_{2}(s-1,i))\Pr({\Oc}_{\text{enc}}(s-1,i,\Rate))\big{)}  \\
                             &\stackrel{(b)}{\leq} B\frac{\delta}{3B}                             
\end{align*}
where $(a)$ follows from the fact that keys used in $s$'th superblock are generated in $s-1$'th superblock, and $(b)$
follows due to the definitions of $E_1(s,b)$ and $E_2(s,b)$. 
Therefore, $\Pr(E(s,b,\delta)|\bar{\Oc}_{\text{enc}}(s,b,\Rate))\leq \delta$.
Finally,  
\begin{align*}
 \Pr(E(s,b,\delta)|\bar{\Oc}_{\text{sec}}(s,b,\Rate,\delta)) &= \Pr(\bar{\Oc}_{\text{enc}}(s,b,\Rate)|\bar{\Oc}_{\text{sec}}(s,b,\Rate,\delta)) \Pr(E(s,b,\delta)|\bar{\Oc}_{\text{enc}}(s,b,\Rate)) \\
								& \quad + \Pr({\Oc}_{\text{enc}}(s,b,\Rate)|\bar{\Oc}_{\text{sec}}(s,b,\Rate,\delta))
\Pr(E(s,b,\delta)|{\Oc}_{\text{enc}}(s,b,\Rate)) \\
                               &\leq  \Pr({\Oc}_{\text{enc}}(s,b,\Rate)|\bar{\Oc}_{\text{sec}}(s,b,\Rate,\delta)) + \Pr(E(s,b,\delta)|\bar{\Oc}_{\text{enc}}(s,b,\Rate))\\
                                 &\leq  \delta 
\end{align*}
where the last inequality follows from the fact that $\Pr({\Oc}_{\text{enc}}(s,b,\Rate)|\bar{\Oc}_{\text{sec}}(s,b,\Rate,\delta))=0$. This concludes the achievability.
Now, we prove the converse. Consider a power allocation function $P(s,b)$, which satisfies
the average power constraint 
\begin{align}
\limsup_{S,B \to \infty} \frac{1}{SB} \sum_{s=1}^S \sum_{b=1}^B P(s,b) \leq P_{\text{avg}} \label{appa:const1}
\end{align}
Let $\delta>0$.
It follows from the converse proof of ergodic secrecy capacity \cite{Secrecy:08}, and law of large numbers
 that $\exists B_1,N_1$ such that for every $S$, $B>B_1$, and $N>N_1$,
 the time-average equivocation rate\footnote{For any reliable code that yields vanishing probability of error 
as $S,B,N \to \infty$.} is bounded as
\begin{align}
\frac{1}{SBN}\sum_{s=1}^S\sum_{b=1}^B  H(W(s,b)|\Zm^{SB},\hv^{SB}) \nonumber\\
\leq \limsup_{S,B \to \infty} \sum_{s=1}^S\sum_{b=1}^B \frac{1}{SB} R_s(s,b) +\delta \label{appa:conv1}
\end{align} 
If $\Rate$ is an $\epsilon$ achievable rate, 
then $\exists B_2,N_2$ such that $\forall B>B_2, N>N_2$
\begin{align}
  \frac{1}{SBN} &\sum_{s=1}^S\sum_{b=1}^B   H(W(s,b)|\Zm^{SB},\hv^{SB}) \nonumber\\
                    &\stackrel{(a)}{\geq} \sum_{s=1}^S\sum_{b=1}^B  
\frac{1}{SB} (\Rate-\delta)\ind(\bar{\Oc}_{\text{sec}}(s,b,\Rate,\delta)) \nonumber\\ 
                    &\stackrel{(b)}{\geq} (\Rate-\delta)(1-\epsilon-\delta) \label{appa:conv2}
\end{align}
where $(a)$ follows directly from the definition of the event 
$\bar{\Oc}_{\text{sec}}(s,b,\Rate,\delta)= \{H(W(s,b)|\Zm^{SB},\hv^{SB})\geq \Rate-\delta\}\cap
\left\{\frac{1}{N}I(\Xm(s,b);\Ym(s,b))\geq \Rate\right\}$, and $(b)$
follows from  applying the secrecy outage constraint \eqref{outage}, and the law of large numbers.

From \eqref{appa:conv1}, \eqref{appa:conv2}, it follows that any $\epsilon$-achievable rate $\Rate$ satisfies
\begin{align}
\Rate \leq {\Rate}^{\ast} = \limsup_{S,B \to \infty} \sum_{s=1}^S\sum_{b=1}^B \frac{1}{SB} R_s(s,b)/(1-\epsilon) 
\end{align}
Since secrecy outage probability has to be satisfied \eqref{outage}, for $(s,b)$, $s\neq 1$ channel
outage probability also has to be satisfied, i.e., 
$\Pr({\Oc}_{\text{ch}}(s,b,\Rate))\leq \epsilon$, which implies
\begin{align}
\Pr(R_m(s,b)< {\Rate}^{\ast}) \leq \epsilon \label{appa:const3}
\end{align}
Since $R_m(s,b)$ and $R_{s}(s,b)$ are both deterministic functions of the power $P(s,b)$ and instantaneous
channel gains $\hv(s,b)$, it follows that the power allocation function that maximizes ${\Rate}^{\ast}$
under the constraints \eqref{appa:const1}, \eqref{appa:const3} is a stationary function of instantaneous channel gains
$\hv(s,b)$. Interchanging the notations $P(s,b)\equiv P(\hv)$, $R_s(s,b) \equiv R_s(\hv,P)$ and $R_m(s,b) \equiv R_m(\hv,P)$,
we can see that 
for any $\epsilon$ achievable secrecy rate, the constraints 
\eqref{SecrecyCapacity}-\eqref{eq:PowerConstraint2} are satisfied, which completes the proof.

\section{Proof of Theorem~\ref{MainCSICapacity}}\label{a:MainCSI}
The proof is very similar to the proof for full CSI, hence we only point out the differences. For full CSI, 
key generation occurs at the end of every \emph{block}, using privacy amplification.  Due to lack of eavesdropper
channel state at the legitimate nodes, this is no longer possible. However, as shown in \cite{Secrecy:08},
it is still possible to generate secret key bits over a \emph{superblock}. 
The following lemma replaces Lemma~\ref{l:privacyamp} in the full CSI case.
\begin{lemma}
Let us define $W_{X}(s) = \{W_{X}(s,b)\}_{b=1}^B$, where
\begin{align*}
W_{X}(s,b)=
\begin{cases}
[W_{sec}(s,b)~ W_{x1}(s,b)],~{\Oc}_{\text{enc}}(s,b,\Rate)\mbox{ does not occur}\\
[W_{x2}(s,b)],~{\Oc}_{\text{enc}}(s,b,\Rate)\mbox{ occurs}
\end{cases} 
\end{align*}
and similarly define $\hat{W}_{X}(s)$. 
There exists $N_2>0,B_1>0$ such that, $\forall N>N_2,B>B_1$, and for any superblock $s$, 
the transmitter and receiver can generate 
secret key bits $V(s) = G(W_X(s))$ and $\hat{V}(s) = G([\hat{W}_{X}(s))$
respectively, 
such that $V(s) = \hat{V}(s)$ if none of the error events $\{E_i(s,b)\}_{b=1}^B$, $i \in \{1,2\}$   occur, and
\begin{align} 
H(V(s)) &= NB(\expect[R_s(\Hm,P)] - \delta) \label{appb:privacyamp1}\\
\frac{1}{NB} I(V(s);\Zm^{SB},\hv^{SB},G) &\leq \delta \label{appb:privacyamp2} 
\end{align}
\end{lemma}
The proof is very similar to the proof of  Lemma~\ref{l:privacyamp},
and is omitted here. Following the same equivocation and error analysis
in the full CSI case, we can see that any rate $\Rate< \expect[R_s(\Hm,P)]/(1-\epsilon)$ 
is achievable. The converse proof is also
the same as in full CSI case, and is omitted here.
\section{Proofs in Section~\ref{s:FullCSIPower}}
\subsection{Proof of Lemma~\ref{bmaxlemma}}\label{a:bmaxlemma}
The parameter ${\Rate}_{\max}$ is the maximum value for which the problem \eqref{eq:PowerControl}-\eqref{serviceconstraint} has a solution; hence
the average power constraint \eqref{powerconstraint} is active. Moreover, the outage constraint \eqref{serviceconstraint}
is also active,
and 
due to the fact that $R_m(\hv,P)$ is a concave increasing function of $P(\hv)$, we have $\Pr(R_m(\Hm,{\Rate}_{\max})={\Rate}_{\max})= (1-\epsilon)$,
since otherwise one can further increase ${\Rate}_{\max}$ to find a power allocation function that
satisfies the equality. Since for a given $\hv$, the power allocation
function that yields ${\Rate}_{\max}$ is $P_{\text{inv}}(\hv,{\Rate}_{\max})$, we have
\begin{align*}
P_{\text{avg}} = \int_{\hv \in \Kc} P_{\text{inv}}(\hv,{\Rate}_{\max})f(\hv)d\hv  
\end{align*}
where $\Kc$ the set of channel gains for which the system operates at rate ${\Rate}_{\max}$, and $\Pr(\Hm \in \Kc)=(1-\epsilon)$.
 The set $\Kc$ contains channel gains $\hv$ for which $P_{\text{inv}}(\hv,{\Rate}_{\max})$ takes minimum values, so that the average power constraint is
 satisfied for the maximum possible ${\Rate}$.
Since $P_{\text{inv}}(\hv,P)=\frac{2^{\Rate}-1}{h_m}$ is a decreasing function of $h_m$, one can see that the choice of $\Kc$ that yields ${\Rate}_{\max}$
is $\Kc=\{ \hv: h_m \geq c \}$. Since the probability density function of $\Hm$ is well defined, $\Pr(H_m=0)=0$, hence $c>0$,
which, along with $P_{\text{avg}}>0$, implies that $\Rate_{\max}>0$.
\subsection{Proof of Lemma~\ref{l:Rb}}\label{al:Rb}
Let ${\Rate}_{\max}>{\Rate}>{\Rate}'>0$. Then, any $P(\hv)$ that satisfies $\Pr(R_m(\Hm,P)<{\Rate})\leq \epsilon$, would also satisfy
$\Pr(R_m(\Hm,P)<{\Rate}')\leq \epsilon$. So, the set of power allocation functions that satisfy
\eqref{serviceconstraint} shrinks as ${\Rate}$ increases, hence $\expect[R_s(\Hm,P^{\Rate})]$ is a non-increasing function of ${\Rate}$.
Now, we prove that $\expect[R_s(\Hm,P^{\Rate})]$ is continuous. 
From Lemma~\ref{l:FullCSIPower}, we know that
\begin{align*}
{P}^{{\Rate}}(\hv) &= P_{\text{wf}}(\hv,\lambda_{\Rate})+ \ind(\hv \in \Gc(\lambda_{\Rate},k_{\Rate})) (P_{\text{inv}}(\hv,{\Rate}')-P_{\text{wf}}(\hv,\lambda_{\Rate}))^+ \\
{P}^{{\Rate}'}(\hv) &= P_{\text{wf}}(\hv,\lambda_{\Rate'})+ \ind(\hv \in \Gc(\lambda_{\Rate'},k_{\Rate'})) (P_{\text{inv}}(\hv,{\Rate'})-P_{\text{wf}}(\hv,\lambda_{\Rate'}))^+
\end{align*}
where $(\lambda_{\Rate},k_{\Rate})$ and $(\lambda_{\Rate'},k_{\Rate'})$ are constants that satisfy
\eqref{optimalPower:3} and \eqref{optimalPower:4} with equality with respect to parameters $\Rate$ and $\Rate'$, respectively.
Let us define another power allocation function $\tilde{P}^{{\Rate}'}(\hv)$ such that
$$\tilde{P}^{{\Rate}'}(\hv) = P_{\text{wf}}(\hv,\lambda_{\Rate})+ \ind(\hv \in \Gc(\lambda_{\Rate},k_{\Rate})) (P_{\text{inv}}(\hv,{\Rate}')-P_{\text{wf}}(\hv,\lambda_{\Rate}))^+ $$
It is easy to see that
$\expect[R_s(\Hm,\tilde{P}^{{\Rate}'})] \leq \expect[R_s(\Hm,P^{{\Rate}'})]$.
Combining the facts i) for any $\hv$, $P_{\text{inv}}(\hv,{\Rate})$ is a 
continuous function of ${\Rate}$ ii) $R_s(\hv,P)$ is a continuous function of $P$
iii) integration preserves continuity, we can see 
that $\int R_s(\hv,P_{\text{inv}}^{\Rate}) - R_s(\hv,P_{\text{inv}}^{{\Rate}'}) 
\ind(\hv \in \Gc(\lambda_{\Rate},k_{\Rate})) f(\hv)d\hv$ is a continuous function 
of ${\Rate}'$. Hence, for any $\gamma>0$, one can find a $\delta>0$ such that for any $\Rate'<\Rate$, $|\Rate'-\Rate|<\delta$,
\begin{align}
\expect[R_s(\Hm,P^{\Rate})]  - \expect[R_s(\Hm,P^{{\Rate}'})] & \leq 
\expect[R_s(\Hm,P^{{\Rate}})]- \expect[R_s(\Hm,\tilde{P}^{{\Rate}'})]  \\                  
      & \leq \int R_s(\hv,P_{\text{inv}}^{\Rate}) - R_s(\hv,P_{\text{inv}}^{{\Rate}'}) \ind(\hv \in \Gc(\lambda_{\Rate},k_{\Rate})) f(\hv)d\hv \nonumber \\
          &\leq \gamma \nonumber
\end{align}
which proves that $\expect[R_s(\Hm,P^{{\Rate}})]$ is a left continuous function of ${\Rate}$. 
Following a similar approach, it can also be shown that $\expect[R_s(\Hm,P^{\Rate})]$ is continuous
from the right.

\subsection{Proof of Lemma~\ref{fullCSIuniquepwr}}\label{a:fullCSIuniquepwr}
If $\expect[R_s(\Hm,P^{\Rate})]|_{\Rate=0}=0$,  then the unique solution of $\Rate = \expect[R_s(\Hm,P^{\Rate})]/(1-\epsilon)$
is $\Rate=0$. So, consider $\expect[R_s(\Hm,P^{\Rate})]|_{\Rate=0}=0$.
 It is easy to see that, $\frac{\expect[R_s(\Hm,P^{{\Rate}_{\max}})]}{{\Rate}_{\max}}\leq(1-\epsilon)$, since
\begin{align*}
\expect[R_s(\Hm,P^{{\Rate}_{\max}})] &= \int_{h_m \geq c} R_s(\hv ,P) f(\hv) d\hv \\
                                   &\leq R_m(\hv ,P)(1-\epsilon) \\                               
                                   &= {\Rate}_{\max} (1-\epsilon)
\end{align*}
follows from definition of parameter $c$, and the inequality $R_s(\hv,P) \leq R_m(\hv,P)$.
Combining the facts that, the function  $\frac{\expect[R_s(\Hm,P^{\Rate})]}{{\Rate}}$ is continuous and strictly decreasing on $(0,{\Rate}_{\max}]$, $\lim_{{\Rate}\to 0^+}\frac{\expect[R_s(\Hm,P^{\Rate})]}{{\Rate}}= \infty$
and $\frac{\expect[R_s(\Hm,P^{{\Rate}_{\max}})]}{{\Rate}_{\max}}\leq(1-\epsilon)$, by the intermediate value theorem,
 there exists a unique $\Rate>0$, which satisfies $\Rate=\expect[R_s(\Hm,P^{\Rate})]/(1-\epsilon)$.
\section{Proof of Lemma~\ref{l:FullCSIPower}}\label{a:FullCSIPower}
We use Lagrangian optimization approach to find $P^{\Rate}(\hv)$. 
We can express $\expect[R_s(\Hm,P^{\Rate})]$ given in \eqref{eq:PowerControl}-\eqref{serviceconstraint} as
\begin{align}
\max_{P(\hv),\Gc}J(P(\Hm)) \nonumber\\
\textrm{s.t}\hspace{0.1cm} R_m(\hv,P) \geq {\Rate}, \quad \forall \hv \in \Gc\nonumber\\
\Pr(\Hm \in \Gc)= 1-\epsilon \label{p2problem}
\end{align}
where the Lagrangian $J(P(\Hm))$ is given by the equation\footnote{Note that we leave the constraint \eqref{serviceconstraint} as is, and
not include it in $J(P(\Hm))$.}
\begin{align}
J(P(\Hm))&=\int R_s(\hv,P)f(\hv)d\hv  \nonumber\\
&-\lambda \left[ \int P(\hv)f(\hv)d\hv -P_{\text{avg}} \right]  \label{Lagrangian2}
\end{align}
Here, $\Gc$ is a set which consists of $\hv$ for which
 $R_m(\hv,P) \geq {\Rate}$ must be satisfied. We will show in this proof that it is of the form \eqref{setG}. 
This problem is identical to \eqref{eq:PowerControl}, since their constraint sets are identical.
Hence solution of this problem would also
yield $P^{\Rate}(\hv)$. 
In the following two-step approach, we proceed to find $P^{\Rate}(\hv)$.
Let us fix $\lambda>0$.
\begin{enumerate}
\item 
For any $\Gc \subseteq [0, \infty) \times [0, \infty)$, we find $P_{\Gc}(\hv)$, which is defined as
\begin{align}
P_{\Gc}(\hv)=\arg\max_{P(\hv)} J(P(\Hm)) \nonumber\\
\textrm{s.t} \hspace{0.5cm} R_m(\hv,P)\geq {\Rate}, \forall \hv \in \Gc \label{objective1}
\end{align} 
\item Using the result of part 1, we find $P^{\Rate}(\hv)$, by finding the set $\Gc$ that maximizes $J'(P(\Hm))$, subject
to a constraint $\Pr(\Hm \in \Gc)= 1- \epsilon$.
\end{enumerate}
We start with step 1. Since both $\lambda$ and ${\Rate}$ are fixed, 
therefore we drop them from $P_{\text{inv}}(\cdot)$ and
$P_{\text{wf}}(\cdot)$, in the following parts to simplify the notation.
\begin{lemma}\label{appc:arbitraryGlemma}
If the problem \eqref{objective1} has a feasible solution, then it could be expressed as 
\begin{equation}
P_{\Gc}(\hv)=P_{\text{wf}}(\hv)+[P_{\text{inv}}(\hv)-P_{\text{wf}}(\hv)]^+\textbf{1}(\hv \in \Gc) \label{objective1soln}
\end{equation}
where $P_{\text{wf}}(\hv)$ and $P_{\text{inv}}(\hv)$ are given in \eqref{Waterfilling} and \eqref{ChannelInversion}, respectively. 

\end{lemma}
\begin{proof}
We will interchangeably use $\hv = [h_m~h_e].$
Due to \eqref{objective1},
$R_m(\hv,P)=\log(1+P(\hv)h_m)\geq {\Rate}$, $\forall \hv \in \Gc$.
Hence, there is a minimum power constraint for set $\Gc$, as
\begin{equation}
P(\hv)\geq P_{\text{inv}}(\hv)=\frac{2^{\Rate}-1}{h_m}  , \forall \hv \in \Gc \label{minpwrsoln}
\end{equation}
Define $\Kc$ as the set in which
the minimum power constraint \eqref{minpwrsoln} is not active, i.e.,
\begin{equation}
\Kc= \left\{  \hv \in \Gc : P(\hv) > P_{\text{inv}}(\hv)
\right\} \cup  \bar{\Gc}  \nonumber
\end{equation}
where $\bar{\Gc}$ is complement of $\Gc$. First, we focus on the solution in the nonboundary set.
Since the optimal solution must satisfy the Euler-Lagrange
equations,
\begin{equation}
\frac{dJ(P(\hv))}{dP(\hv)}=0, \hv \in \Kc \nonumber
\end{equation}
For $\hv \in \Kc$, we get the following condition
\begin{equation}
\frac{h_m}{1+h_mP(\hv)}-\frac{h_e}{1+h_eP(\hv)}-\lambda=0  \nonumber
\end{equation}
whose solution yields 
\begin{align}
P(\hv)=
\frac{1}{2}\left[\sqrt{\left(\frac{1}{h_e}-\frac{1}{h_m}\right)^2+\frac{4}{\lambda}\left(\frac{1}{h_e}-\frac{1}{h_m}\right)}
-\left(\frac{1}{h_e}+\frac{1}{h_m}\right)\right]\nonumber
\end{align}
If for some $\hv \in \Kc$, the value $P(\hv)$ is negative, then
due to the concavity of $J(P(\hv))$ with respect to $P(\hv)$, the optimal 
value of $P(\hv)$ is zero \cite{Secrecy:08}. Therefore, the solution yields
\begin{align}
P(\hv)=P_{\text{wf}}(\hv), \quad \forall \hv \in \Kc
\end{align}
Combining the result with the minimum power constraint inside set $\Gc$,
the solution of \eqref{objective1} yields
%
\eqref{objective1soln}, which concludes the proof.
\end{proof} 
Now, we find $P^{\Rate}(\hv)$.
We proceed by further simplifying the Lagrangian in \eqref{Lagrangian2},
for the case where $P(\hv)=P_{\Gc}(\hv)$, for a given $\Gc$ as follows.
\begin{align}
J(P_{\Gc}(\Hm))&= \quad \int_{\hv  \in \Gc}\left[ R_s(\hv,P)-\lambda P(\hv) \right]f(\hv)d\hv \nonumber\\   
         & \quad + \int_{\hv \notin \Gc}\left[ R_s(\hv,P)-\lambda P(\hv) \right]f(\hv)d\hv \nonumber\\   
         & = \quad \int \left[ R_s(\hv,P_{\text{wf}})-\lambda P_{\text{wf}}(\hv) \right]f(\hv)d\hv \nonumber\\ 
         & \quad+ \int_{\Gc} \left\{ \left[R_s(\hv,P_{\text{inv}})-R_s(\hv,P_{\text{wf}})\right]^+ \right.\nonumber\\
         & \qquad-\lambda \left.\left[P_{\text{inv}}(\hv)-P_{\text{wf}}(\hv)\right]^+\right\} f(\hv)d\hv\label{Jeqnfinal2}  
\end{align}
After this simplification, the first term in \eqref{Jeqnfinal2} does not depend on $\Gc$.
We conclude the proof by showing that 
$P^{\Rate}(\hv) = P_{\Gc^{\ast}}(\hv)$
where the set $\Gc^{\ast}$ is defined as follows, 
\begin{align}
\Gc^{\ast}=\left\{ \hv: \left[R_s(\hv,P_{\text{inv}})-R_s(\hv,P_{\text{wf}})\right]^+ 
-\lambda \left[P_{\text{inv}}(\hv)-P_{\text{wf}}(\hv)\right]^+ \geq k \right\}\label{optimalG}
\end{align}
where the parameter $k$ is a constant that satisfies $\Pr (\Hm \in \Gc^{\ast}) = (1- \epsilon)$.
We prove this by contradiction. First define 
$\xi(\hv)= \left[R_s(\hv,P_{\text{inv}})-R_s(\hv,P_{\text{wf}})\right]^+ -\lambda \left[P_{\text{inv}}(\hv)-P_{\text{wf}}(\hv)\right]^+$. 
Then,  it follows from \eqref{Jeqnfinal2} that $\Gc^{\ast}$ is the set that maximize \eqref{Jeqnfinal2}, so
\begin{equation}
\Gc^{\ast}= \arg\max_{\Gc} \int_{\Gc} \xi(\hv) f(\hv)d\hv   \nonumber
\end{equation}
Assume that some other $\Gc' \neq \Gc^{\ast}$ is optimal, where $\Pr(\Hm \in \Gc')=1-\epsilon$. However, we have
\begin{align}
&\hspace{1cm} J(P_{\Gc^{\ast}}(\Hm))-J(P_{\Gc'}(\Hm)) \nonumber\\
&=\int_{\Gc^{\ast}}\xi(\hv)f(\hv)d\hv-\int_{\Gc'}\xi(\hv)f(\hv)d\hv \nonumber\\
&=\int_{\Gc^{\ast} \char`\\ \Gc'}\xi(\hv)f(\hv)d\hv-\int_{\Gc' \char`\\ \Gc^{\ast}}\xi(\hv)f(\hv)d\hv \nonumber\\
&\geq 0
\end{align} 
since
\begin{equation}
\int_{\Gc^{\ast} \char`\\ \Gc'}f(\hv)d\hv=\int_{\Gc' \char`\\ \Gc^{\ast}}f(\hv)d\hv \nonumber
\end{equation}
and 
\begin{equation}
\xi(\hv)|_{\hv \in \Gc^{\ast}} \geq \xi(\hv)|_{\hv \in \Gc'},\hspace{0.2cm}\forall \hv \nonumber
\end{equation}
by definition. This contradicts our assumption that $\Gc'$ is optimal.
Note that, $\Gc^{\ast}$ is identical to \eqref{optimalPower}.
This concludes the proof.

\section{Proof of Lemma~\ref{t:MainCSIPower}}\label{a:MainCSIPower}
The proof goes along similar lines as in Appendix~\ref{a:FullCSIPower}, so we skip the details here. 
We solve the problem for a fixed $\lambda > 0$.
First, for any given  
$\Gc \in [0,\infty)$, we define the following problem, the solution of which yields $P_{\Gc}(h_m)$.
\begin{align}
P_{\Gc}(h_m) & = \arg\max_{P(h_m)} J(P(H_m))  \label{appd:generalpower} \\
          & \mbox{subject to: } R_m([h_m,h_e],P) \geq {\Rate}, \forall h_m \in \Gc \label{appd:minrateconst} 
\end{align}
\begin{lemma}\label{appd:GeneralPowerLemma}
If the problem \eqref{appd:generalpower} has a feasible solution, then it can be expressed as
\begin{align}
P_{\Gc}(h_m) = P_w(h_m,\lambda) + \ind(h_m \in \Gc) \left(  P_{\text{inv}}(h_m,{\Rate}) - P_{w}(h_m,\lambda)  \right)^{+} \label{appd:generalsoln}
\end{align}
\end{lemma}
\begin{proof}
The proof uses the same approach as in proof of Lemma~\ref{appc:arbitraryGlemma}. We define the set $\Kc$ such that for any $h_m \in \Kc$,
the minimum rate constraint in \eqref{appd:minrateconst} is not active. Since the optimal solution must satisfy the Euler Lagrange
equations, we have
$$ \frac{dJ(P(h_m))}{dP(h_m)}=0,\mbox{ }h_m \in \Kc $$ 
If we solve the equation for any given $h_m$, we get
\begin{align}
\frac{h_m \Pr (H_e \leq h_m)}{1+h_mP(h_m)}-\int_0^{h_m}\left( \frac{h_e}{1+h_eP(h_m)} \right)f(h_e)dh_e - \lambda = 0
\end{align}
If the power allocation function that solves the equation is negative, then by the convexity of the objective function  \cite{Secrecy:08},
the optimal value of $P(h_m)$ is $0$. Hence, we get $P_w(\hv,\lambda)$ as the resulting power allocation function. 
Whenever the minimum rate constraint \eqref{mainserviceconstraint} is active, 
we get the channel inversion power allocation function, $P_{\text{inv}}(\hv,{\Rate})$.
\end{proof}
Now, using Lemma~\ref{appd:GeneralPowerLemma}, we solve the following problem,
\begin{align}
\max_{P(h_m),\Gc}J(P(H_m)) \label{appd:lemma2}\\
\textrm{s.t}\hspace{0.1cm} R_m(\hv,P) \geq {\Rate}, \quad \forall \hv \in \Gc\nonumber\\
\Pr(H_m \in \Gc)= 1-\epsilon \nonumber
\end{align}
the solution of which yields $P^{\Rate}(h_m)$. Lemma~\ref{appd:GeneralPowerLemma} proves that the solution is a time-sharing between
$P_w(h_m,\lambda)$ and $P_{\text{inv}}(h_m,{\Rate})$. Now, we find the optimal $\Gc$.
\begin{lemma}
The solution of \eqref{appd:lemma2} is of the form \eqref{appd:generalsoln}, with the set $\Gc^{\ast} = [c, \infty)$, where
$c$ is a constant which solves $\Pr(H_m \geq c)=1-\epsilon$.
\end{lemma}
\begin{proof}
Let $P_{{\Gc}^{\ast}}(h_m)$ and $P_{{\Gc}'}(h_m)$ be the power allocation
functions that are solutions of \eqref{appd:generalsoln} given the sets ${\Gc}^{\ast}$ and $\Gc'$, respectively.
We show that, any choice of $\Gc' \neq \Gc^{\ast}$,
such that $\Pr(H_m \in \Gc') = 1-\epsilon$ is suboptimal, i.e.,  
\begin{align*}
J(P_{{\Gc}^{\ast}}(H_m)) - J(P_{{\Gc}'}(H_m)) \geq 0
\end{align*}
We continue as follows.
\begin{align*}
&J(P_{{\Gc}^{\ast}}(H_m)) - J(P_{{\Gc}'}(H_m)) = \\
& \int\left\{  \int_{h_m \in \Gc^{\ast}} \left( \left[R_s([h_m,h_e],P_{\text{inv}})-R_s([h_m,h_e],P_{w})\right]^+ 
-\lambda \left[P_{\text{inv}}(h_m,{\Rate})-P_{\text{wf}}(h_m,\lambda)\right]^+ \right) f(h_m)dh_m \right\} f(h_e) dh_e \\
&- \int \left\{ \int_{h_m \in \Gc'}\left( \left[R_s([h_m,h_e],P_{\text{inv}})-R_s([h_m,h_e],P_{w})\right]^+ 
- \lambda \left[P_{\text{inv}}(h_m,{\Rate})-P_{\text{wf}}(h_m,\lambda)\right]^+ \right) f(h_m)dh_m \right\} f(h_e) dh_e
\end{align*}
Note that,
for any $h_m' \in \Gc^{\ast} \backslash \Gc'$ and $ h_m'' \in \Gc' \backslash \Gc^{\ast}$, we have $h_m' > h_m''$. 
Since $P_{w}(h_m',\lambda)\geq P_{w}(h_m'',\lambda)$
and $P_{\text{inv}}(h_m',\lambda) < P_{\text{inv}}(h_m'',\lambda)$, we have 
$$[P_{\text{inv}}(h_m',{\Rate})-P_{w}(h_m',\lambda)]^+ \leq \left[P_{\text{inv}}(h_m'',{\Rate})-P_{w}(h_m'',\lambda)\right]^+$$
Since $R_s(\cdot,P)$ is a concave increasing function of $P(\cdot)$ \cite{Secrecy:08},
and for $P_{w}(\cdot,P)$, we have $\frac{d P_{w}(\cdot,P)}{dP}=\lambda$. 
Therefore for any $h_e$, we have
\begin{align*}
&\left[R_s([h_m',h_e],P_{\text{inv}})-R_s([h_m',h_e],P_{w})\right]^+ 
-\lambda \left[P_{\text{inv}}(h_m',{\Rate})-P_{w}(h_m',\lambda)\right]^+ \\
-&\left[R_s([h_m'',h_e],P_{\text{inv}})-R_s([h_m'',h_e],P_{w})\right]^+ 
+\lambda \left[P_{\text{inv}}(h_m'',{\Rate})-P_{w}(h_m'',\lambda)\right]^+ \geq 0
\end{align*}
Combining this result with the packing arguments following \eqref{optimalG} in Appendix~\ref{a:FullCSIPower}, we get
\begin{align*}
J(P_{{\Gc}^{\ast}}(H_m)) - J(P_{{\Gc}'}(H_m)) \geq 0
\end{align*}
hence concluding the proof. Note that, this result can also be proved using
the arguments of Section 4 in \cite{Luo:03}. 
\end{proof}
\section{Proofs in Section~\ref{s:finitebuffer}}
\subsection{Proof of Lemma~\ref{l:stationarity}} \label{app:stationarity}
Due to Theorem 1.2 of Section VI in \cite{As87}, it suffices to show
that $Q_M(t)$ is a positive recurrent regenerative process.
Note that $Q_M(t)$ is a Markov process with an uncountable state space $[0 ~ M]$, since
$Q_M(t)$ can be written as
$Q_M(t+1) = \min( M, Q_M(t)+ R_s(t) - \ind(\bar{\Oc}_{\text{x}}(t) \cap \bar{\Oc}_{\text{key}}(t))) $
where $R_s(t)$ and $\bar{\Oc}_{\text{x}}(t)$ are i.i.d., and 
$\bar{\Oc}_{\text{key}}(t)= \big{\{} Q_M(t) + R_s(t) -\Rate \geq 0  \big{\}}$
depends only on  $Q_M(t)$ and $R_s(t)$. Therefore, $Q_M(t+1)$ is independent of $\{Q_M(i)\}_{i=1}^{t-1}$ given $Q_M(t)$,
hence Markovity follows. Now, we prove that $Q_M(t)$ is a recurrent regenerative process
where regeneration occurs at times $t_1,t_2,\cdots$ such that $Q_M(t_i) = M$.
A sufficient condition for this is to show that $Q_M(t)$ has an accessible atom \cite{Markovian}.
\begin{definition}
An accessible atom $M$ is a state that is hit with positive probability starting from any state,
i.e., $\sum_{t=1}^{\infty}\Pr(Q_M(t)=M| Q_M(1) = i) > 0 $ $\forall i$.
\end{definition}
\begin{lemma}\label{l:atom}
$Q_M(t)$ has an accessible atom $M$.
\end{lemma}
\begin{proof}
Assume $Q_M(1)= i$, $i \in [0,M]$. 
Note that, $R_s(t)$ and $\Oc_{\text{x}}(t)$ are both i.i.d.
Also note that, $\Pr\big{(} R_s(t)-\Rate \ind(\bar{\Oc}_{\text{x}}(t)) >0 \big{)} >0 $ $\forall t$ \footnote{Since considering otherwise
would lead to the uninteresting scenario where there are no buffer overflows (since the key queue cannot grow),
hence any buffer size $M>C_{F}^{\epsilon'}$ is sufficient to achieve $\epsilon'$ secrecy outage probability.}.
Find $\gamma >0$ such that $\Pr\big{(}R_s(t)-\Rate \ind(\bar{O}_{\text{x}}(t)) >\gamma \big{)} = \gamma $ $\forall t$.
Let $\eta_i = \lceil \frac{M-i}{\gamma} \rceil $. Then,
\begin{align*}
\Pr(Q_M(\eta_i+1)=M| Q_M(1) = i) &\geq \prod_{t=1}^{\eta_{i}} \Pr\big{(}R_s(t)+ \ind(\bar{\Oc}_{\text{x}}(t)) >\delta \big{)} \\
                           &\geq  \gamma^{\eta_i}  \\
                           & >0
\end{align*}
\end{proof}
Since $Q_M(t)$ is a regenerative process, we know that $t_2-t_1,t_3-t_2,\cdots$ are i.i.d. random variables. Define a random variable $\tau$, with distribution identical to $t_{i+1}-t_{i}$.
Now we show that $Q_M(t)$ is positive recurrent, by showing 
$\expect[\tau]< \infty$.
Consider another recursion 
\begin{align}
 Q'_M(t+1)= \min \big{(} M,Q'_M(t)+R_s(t)- \Rate\ind(\bar{\Oc}_{\text{x}}(t)) \big{)}^+ \label{QpDynamics}
\end{align} 
with $Q'_M(1) = Q_M(1)$. It is clear that $Q'_M(t)$ is also regenerative, where regeneration occurs
at $\{ t_i' \}$, where $Q_M(t_i')=M$, and let $\tau'$ be equal in distribution to $t_{i+1}'-t_{i}'$.
\begin{lemma}\label{Tlemma}
\begin{align*}
\expect[\tau] \leq \expect[\tau']
\end{align*}
\end{lemma}
\begin{proof}
It suffices to show that when $Q_M(t)\neq M$, $Q'_M(t) \leq Q_M(t)$.
By induction,  assuming $Q_M'(t)\leq Q_M(t)$, we need to verify
that $ Q_M'(t+1)\leq Q_M(t+1)$.
Consider $Q_M(t+1)<M$. Then,
\begin{align}
 Q_M(t+1) &={\Big (} Q_M(t) + R_s(t) - {\Rate} \ind( \bar{\Oc}_{\text{x}}(t) \cap \bar{\Oc}_{\text{key}}(t) {\Big )}^+\nonumber\\
&\geq  \left(Q_M(t) + R_s(t)-{\Rate} \ind(\bar{\Oc}_{\text{x}}(t)) \right)^+ \nonumber\\
&\geq \left(Q'_M(t) + R_s(t)-{\Rate} \ind(\bar{\Oc}_{\text{x}}(t)) \right)^+
\nonumber\\
&=Q'_M(t+1) \nonumber
\end{align}
\end{proof}
Note that $Q'_M(t)$ is regenerative both at states $0$ and $M$.
Let $\expect[\tau'_1]$ denote the expected time for the process $Q'_M(t)$ 
to hit $0$ from $M$, and $\expect[\tau'_2]$ denote
the expected time to hit $M$ from $0$. Then,
\begin{align}
\expect[\tau'] \leq \expect[\tau'_1] + \expect[\tau'_2] \label{Tbar}
\end{align}
Since the key queue has a negative drift, i.e.,
 $\mu_{\Rate}= \expect[R_s(\Hm,P) - \Rate\ind(\bar{\Oc}_{\text{x}}(t))]<0$,
it is clear that 
$\expect[\tau'_1]< \infty$. 
Now, we show that $\expect[\tau'_2]< \infty$. Following the approach of Lemma~\ref{l:atom}, find $\gamma>0$ such that
 $\Pr\big{(} R_s(t)-\Rate \ind(\bar{\Oc}_{\text{x}}(t)) >\gamma \big{)} = \gamma $ $\forall t$.
Let $\eta = \lceil M/\gamma \rceil$. Then,
$\Pr(Q_M(\eta+1)=M | Q_M(1) = 0 )\geq \gamma^{\eta}>0$,
and
\begin{align*}
\expect[\tau'_2]&\leq \sum_{i=0}^{\infty}(\eta + i(\expect[\tau'_1]+\eta))\gamma^{\eta}(1-\gamma^{\eta})^i  \\
            & \leq  \eta \gamma^{\eta}\sum_{i=0}^{\infty}(1-\gamma^{\eta})^i + \sum_{i=0}^{\infty}(1-\gamma^{\eta})^i i(\expect[\tau'_1]+\gamma^{\eta}) \\
            & < \infty
\end{align*}
The first inequality follows from the fact that with probability $\gamma^{\eta}$, $Q_M(t)$ hits $M$ at $\eta$'th block 
and with probability $(1-\gamma^{\eta})$, key queue goes back to state $0$ at $(\expect[\tau'_1]+\gamma^{\eta})$'th block (on average).
The last inequality follows from $0<\gamma^{\eta}<1$, and ratio test. This result, along with \eqref{Tbar} and Lemma~\ref{Tlemma} concludes that
$Q_M(t)$ is a positive recurrent regenerative process, which concludes the proof.

\subsection{Proof of Lemma~\ref{t:lossPr}} \label{app:BufferOverflowProof}
We follow an indirect approach to prove the lemma.
Let $\{Q(t)\}_{t=1}^{\infty}$ denote the key queue dynamics
of the same system for the infinite buffer case ($M=\infty$).  
First, we use the heavy traffic results in \cite{Kingman62} 
 to calculate the overflow probability of the infinite buffer queue.
Then, we relate the overflow probability
of infinite buffer system to the loss ratio of the finite buffer queue.
The dynamics of the infinite buffer queue is characterized by
\begin{align}
Q(t+1)= Q(t)+R_s(t)-\ind(\bar{\Oc}_{\text{enc}}(t))\Rate \label{QkDynamics}
\end{align}  
where $Q(1)=0$. 
The heavy traffic results we will use are for queues that have a stationary distribution.
Since it is not clear whether $Q(t)$ is stationary or not, we will upper bound
$Q(t)$ by another stationary process $Q'(t)$, and the buffer overflow probability
result we will get for $Q'(t)$ will serve as an upper bound 
for $Q(t)$.

Let $\{Q'(t)\}_{t\geq 1}$ be the process that satisfies the
following recursion
\begin{align}\label{eq:newQstar}
Q'(t+1) = \left(Q'(t) + R_s(t)-{\Rate} \ind(\bar{\Oc}_{\text{x}}(t)
\right)^+
\end{align}
 with $Q'(1) = 0$.
First, we relate $Q'(t)$ to $Q(t)$.
\begin{lemma}\label{lemma:Qstar}
\begin{align}\label{eq:sanwich}
Q(t) \leq Q'(t)+{\Rate},~ \forall t
\end{align}
\end{lemma}
\begin{proof}
Assuming $Q(t)\leq Q'(t)+{\Rate}$, we need to show by induction that $ Q(t+1)\leq
Q'(t+1)+{\Rate}$. There are two different scenarios.
\begin{enumerate}
\item If $Q'(t)+R_s(t)-{\Rate}\ind \left( \bar{\Oc}_{\text{x}}(t) \right)\geq 0$, then,
using the facts $\bar{\Oc}_{\text{enc}}(t) = \bar{\Oc}_{\text{x}}(t) \cap  \bar{\Oc}_{\text{key}}(t)$
and
 $Q'(t)\leq Q(t)$, we obtain
\begin{align}
Q(t)  +R_s(t)  -  {\Rate} \ind \left( \bar{\Oc}_{\text{enc}}(t)\right) 
&\geq Q'(t) + R_s(t)
 - {\Rate} \ind \left(  \bar{\Oc}_{\text{x}}(t) \right) \nonumber\\
&\geq 0 \nonumber
\end{align}
which, using the described key queue recursions in \eqref{QkDynamics}, implies
\begin{align}\label{eq:queueBound1}
Q(t+1) = Q(t) + R_s(t)  
 - {\Rate} \ind \left(  \bar{\Oc}_{\text{x}}(t) \right)
\end{align}
Observe that, by (\ref{eq:newQstar}),
\begin{align}
&Q'(t+1)= Q'(t) + R_s(t)- {\Rate} \ind \left( \bar{\Oc}_{\text{x}}(t)
 \right) \nonumber
\end{align}
which, in conjunction with (\ref{eq:queueBound1}) and $Q(t)\leq
Q'(t)+{\Rate}$, yields  $Q(t+1)\leq Q'(t+1)+{\Rate}$.
\item If $Q'(t)+R_s(t)-{\Rate}\ind \left(  \bar{\Oc}_{\text{x}}(t) \right)<
0$, then $Q'(t+1)=0$.
We further consider two cases. First, if $Q(t)+R_s(t)-{\Rate} \geq 0$,
then,
\begin{align}\label{eq:queueBound3}
Q(t+1)&={\Big (} Q(t) + R_s(t)  - {\Rate} \ind \left(\bar{\Oc}_{\text{x}}(t)\right) {\Big )}^+ \nonumber\\
&\leq {\Big (} Q'(t) +{\Rate} + R_s(t)  - {\Rate} \ind \left( \bar{\Oc}_{\text{x}}(t)\right) {\Big )}^+ 
\leq {\Rate} \nonumber\\
&= Q'(t+1)+{\Rate}
\end{align}
Next, if $Q(t)+R_s(t)-{\Rate} < 0$, then $$Q(t+1)=Q(t)+R_s(t) <
{\Rate}=Q'(t+1)+{\Rate}$$ which,  combined with (\ref{eq:queueBound3}),
yields $$Q(t+1) \leq Q'(t+1)+{\Rate}$$
\end{enumerate}
\end{proof}
Now, we show that $Q'(t)$ converges in distribution to an almost surely finite random variable $Q'$.
First, we need to show that the expected drift of $Q'(t)$ is negative. It is clear from 
\eqref{eq:newQstar} that
the expected drift of the process $Q'(t)$ is equal to $\mu_{\Rate} = \expect[R_s(\Hm,P^R)]-\Rate(1-\epsilon)$.
\begin{lemma} \label{muLemma}
For ${\Rate}>C_F^{\epsilon}$, we have ${\mu}_{\Rate}<0$, and ${\mu}_{\Rate}$ is a
 continuous decreasing function of ${\Rate}$.
\end{lemma}
\begin{proof}
From Lemma~\ref{l:Rb} in Section~\ref{s:FullCSIPower}, we know that $\expect[R_s(\Hm,P^{\Rate})]$ is
a non-increasing continuous function of ${\Rate}$. 
Therefore, $\mu_{\Rate}$ it is a continuous  function of ${\Rate}$.
Furthermore, by definition of $C_F^{\epsilon}$ in \eqref{SecrecyCapacity}, $\mu_{C_F^{\epsilon}}=0$.
Combining these two facts,
we conclude that $\mu_{\Rate}<0$, for ${\Rate}>C_F^{\epsilon}$. 
\end{proof}
\begin{lemma}\label{l:stablelemma}
There exists an almost surely finite random variable
$Q'$ such that, for all $x$,
\begin{equation}\label{eq:stable1}
 \limsup_{t\to\infty} \Pr(Q(t)>x) \leq \Pr(Q'+{\Rate} >x)
\end{equation}
\end{lemma}
 \begin{proof}
Combining Lemma~\ref{muLemma} with the classic results by Loynes \cite{LOY62},
we can see that $Q'(t)$ converges in distribution to an almost surely
finite random variable $Q'$ such that
$$
 \lim_{t\to\infty} \Pr(Q'(t)>x) = \Pr(Q'>x)
$$
Using \eqref{eq:sanwich}, we finish the
proof of the lemma.
\end{proof}
Now, we characterize the tail distribution of the key queue.
\begin{lemma}\label{l:bufferOverflow}
For any given $M\geq 0$, 
\begin{align}
 \lim_{\Rate \searrow C_F^{\epsilon}}  \limsup_{t \to \infty} \Pr\left( \frac{|\mu_{\Rate}| (Q(t)-\Rate)}{{\sigma}_{\Rate}^2}>M \right) \leq
 e^{-2M}\label{bufferOverflow}
\end{align}
\end{lemma} 
\begin{proof}
First, we prove that
\begin{align}
 \lim_{{\Rate} \searrow C_F^{\epsilon}} \Pr\left( \frac{|\mu_{{\Rate}}| Q'}{\sigma_{{\Rate}}^2}>y \right) =
 e^{-2y}, \label{q'bufferOverflow}
\end{align}
which is based on the heavy traffic limit for queues developed in \cite{Kingman62},
see also Theorem 7.1 in \cite{As87}.
In order to prove \eqref{q'bufferOverflow}, we only need to verify the following
three conditions: 
i) $\lim_{{\Rate} \searrow C_F^{\epsilon}}\mu_{{\Rate}}=0$; ii) $\lim_{{\Rate}
\searrow C_F^{\epsilon}}\sigma_{{\Rate}}^2 > 0$; and iii) the set 
$\left\{\left(R_s(\Hm,P^{\Rate})-{\Rate}\ind(\bar{\Oc}_{\text{x}}(t))\right)^2\right\}$
of random
variables 
indexed by ${\Rate}$ is uniformly integrable. 

i) From Lemma~\ref{muLemma}, we obtain
$\lim_{{\Rate} \searrow C_F^{\epsilon}}\mu_{{\Rate}} = 0$.

ii) 
Since $R_s(\Hm,P^{\ast})-C_F^{\epsilon}\bar{\Oc}_{\text{x}}(t)$ is not a
constant random variable, almost surely
\begin{align}
\lim_{{\Rate} \searrow C_F^{\epsilon}}\sigma^2_{{\Rate}}&= \var
[R_s(\Hm,P^{\ast})-C_F^{\epsilon}(\bar{\Oc}_{\text{x}}(t))]>0 \nonumber
\end{align}

iii) Note that, ${\Rate}$ lies on the interval $[0 ~ {\Rate}_{\max}]$, where ${\Rate}_{\max}$,  defined in Lemma~\ref{bmaxlemma}
then we have
\begin{align*} 
 \left(R_s(\Hm,P^{\Rate})-{\Rate}\ind(\bar{\Oc}_{\text{x}}(t))\right)^2 = & R_s(\Hm,P^{{\Rate}})^2 - 2 R_s(\Hm,P^{\Rate}) {\Rate} \ind\left( \bar{\Oc}_{\text{x}}(t)\right) + \\
& {\Rate}^2 \ind\left( \bar{\Oc}_{\text{x}}(t)\right)   \\
    \leq & R_s(\Hm,P^{\Rate})^2 + {\Rate}_{\max}^2
\end{align*}
Since $R_s(\hv,P)$ is a continuous function of $P(\hv)$, and for any ${\Rate}$ on 
the interval $[0 ~ {\Rate}_{\max}]$, $\lim_{c\to\infty}\Pr(P^{\Rate}(\Hm)>c) = 0$, 
hence we can see that $\lim_{c \to \infty} \Pr(R_s(\Hm,P^{\Rate})>c) = 0$. Therefore,
this class of random
variables is uniformly integrable.
This completes the proof of \eqref{q'bufferOverflow}.
This result, in conjunction with  Lemma~\ref{l:stablelemma} completes the
proof.
\end{proof}
Using Lemma 1 in \cite{Ness}, we relate the loss ratio of our finite
buffer queue $Q_M(t)$ to the overflow probability of the infinite buffer queue $Q(t)$
as follows
\begin{align} 
\expect[R_s(\Hm,P^{\Rate})] \limsup_{T \to \infty} L^T(M) \leq \int_{x=M}^{\infty} \limsup_{t\to\infty}\Pr(Q(t)>x) dx   \label{LossOverflow}
\end{align}
Combining Lemma~\ref{l:bufferOverflow} with \eqref{LossOverflow}, the proof is complete.

\end{appendices}


\end{document}